\documentclass{article}

\usepackage{arxiv}

\usepackage[utf8]{inputenc}
\usepackage[T1]{fontenc}
\usepackage{amsmath,amssymb,amsfonts,amsthm}
\usepackage{mathrsfs}
\usepackage{bm}
\usepackage{graphicx}
\usepackage{subcaption}
\usepackage{multirow}
\usepackage{booktabs}
\usepackage[title]{appendix}
\usepackage{xcolor}
\usepackage{textcomp}
\usepackage{algorithm}
\usepackage{algorithmicx}
\usepackage{algpseudocode}
\usepackage{listings}
\usepackage{microtype}
\usepackage{placeins}
\usepackage{url}
\usepackage[hidelinks]{hyperref}
\usepackage{natbib}
\usepackage{doi}
\usepackage{orcidlink}

\DeclareMathOperator{\E}{\mathbb{E}}

\newcommand{\IG}{\operatorname{IG}}
\newcommand{\GIG}{\operatorname{GIG}}
\newcommand{\GHSt}{\operatorname{GHSt}}

\theoremstyle{plain}
\newtheorem{theorem}{Theorem}
\newtheorem{proposition}[theorem]{Proposition}
\newtheorem{corollary}[theorem]{Corollary}
\theoremstyle{definition}

\theoremstyle{remark}

\title{Deep Skew-$t$ Mixture Models}

\author{
Jinran Wu\\
The University of Queensland, Australia\\
\texttt{jinran.wu@uq.edu.au}
\And
You-Gan Wang\\
Guangdong University of Finance and Economics, China\\
\texttt{you-gan.wang@uq.edu.au}
\And
Geoffrey J. McLachlan\\
The University of Queensland, Australia\\
\texttt{g.mclachlan@uq.edu.au}
}

\renewcommand{\shorttitle}{Deep Skew-$t$ Mixture Models}

\begin{document}
\maketitle

\begin{abstract}
High-dimensional clustering is challenging when component distributions are both heavy-tailed and directionally asymmetric. We propose a deep skew-$t$ mixture model (DStMM), a hierarchical factor-analytic mixture based on the generalised-hyperbolic skew-$t$ normal mean--variance representation. A shared inverse-gamma mixing variable is propagated along each complete latent pathway, allowing heavy tails and directional asymmetry to be modelled jointly while preserving conditional Gaussianity. Each complete pathway therefore admits an exact GHST marginal representation. We formalise the reductions to symmetric deep $t$, Gaussian deep-mixture, and single-layer GHST factor-analytic models, discuss local non-identifiability and the implementation-level parameter-counting convention, and derive the conditional generalised inverse Gaussian law used for estimation. Estimation is carried out by a stochastic/Monte Carlo EM algorithm, with an explicit implementation-based parameter count for BIC architecture comparison. Simulation studies show that DStMM performs similarly to the symmetric robust model when skewness is absent but provides increasing gains as directional asymmetry becomes stronger, particularly under heavier tails; the same qualitative behaviour persists under smaller samples and unequal mixture proportions. Two real-data applications provide complementary evidence. On the UCI handwritten-digit benchmark, DStMM gives the strongest clustering performance under a common deep architecture, while on the Gas Sensor Array Drift data, DStMM improves on both deep Gaussian and deep $t$ alternatives and, under the implemented BIC criterion, selects a non-trivial second mixture layer. Together, these results support the value of propagating skewness and heavy-tail variation through a deep latent mixture while retaining an exact pathway-level likelihood.
\end{abstract}


\section{Introduction}
\label{sec:introduction}

Distributional shape is central to model-based clustering. Gaussian mixtures provide a convenient framework for heterogeneous data, but their component geometry is necessarily symmetric, and their likelihood can be sensitive to extreme observations. These limitations are especially important in high dimensions, where heavy tails can distort parameter estimation and directional asymmetry can encourage a symmetric model to represent one heterogeneous group by several fitted components. At the same time, unrestricted component covariance matrices quickly become expensive to estimate. A useful high-dimensional clustering model therefore needs to control covariance complexity while allowing departures from both Gaussian tail behaviour and elliptical symmetry.

Factor-analytic mixtures address the covariance problem by representing within-component dependence through lower-dimensional latent variables \citep{fokoue2003mixtures,mclachlan2003modelling}. The deep Gaussian mixture model (DGMM) of \citet{viroli2019deep} extends this idea by recursively composing factor-analytic mixture transitions, producing a hierarchy of progressively lower-dimensional latent states. This construction combines dimension reduction, local parameter sharing, and flexible marginal structure. Subsequent work has developed variational methods and extensions for mixed observations, missing data, topic models, and related deep probabilistic mixture settings \citep{kock2022variational, fuchs2022mixed, fuchs2022mi2ami, viroli2021deep, hamalainen2020deep, selosse2020bumpy}. Despite this flexibility, the Gaussian transitions underlying the DGMM remain vulnerable when the data contain substantial tail inflation or asymmetric cluster shapes. Robust mixture modelling addresses the first of these departures by replacing Gaussian components with multivariate Student-$t$ distributions. The corresponding scale-mixture representation down-weights observations with large Mahalanobis distances and has motivated robust mixture and factor-analytic models for high-dimensional data \citep{mclachlan1998robust, peel2000robust, mclachlan2007extension, baek2011mixtures, wang2022robust}. The robust deep mixture model (RDMM) of \citet{wu2026rdmm} carries this mechanism into a deep hierarchy by sharing one latent scale across all transitions on an observation's complete pathway. The resulting model substantially improves robustness to heavy tails, but its pathway components remain elliptically symmetric. Thus RDMM addresses tail weight without directly representing directional asymmetry.

A complementary literature has developed skew-normal and skew-$t$ mixture models for clusters exhibiting asymmetric component shapes; see \citet{lee2022overview} for an overview of skew distributions in model-based clustering. In particular, restricted skew-$t$ constructions have been used for robust factor analysis and mixtures of factor analysers \citep{lin2015robust, lin2018robust}, while the canonical fundamental skew-$t$ framework provides a broader formulation that unifies restricted and unrestricted skew-$t$ mixture models \citep{lee2016finite}. Related work has considered efficient EM-type estimation for multivariate skew-normal and skew-$t$ mixtures \citep{lee2018block}, as well as the role of skewness in latent factors and/or error terms in skew factor models \citep{lee2021formulations}. More generally, mixtures of factor analysers based on scale mixtures of fundamental skew-normal distributions provide a flexible framework for simultaneously accommodating dimension reduction, asymmetry, and heavy-tailed behaviour \citep{lee2021scale}. These developments demonstrate the usefulness of skewed factor-analytic components, but they are formulated within conventional single-layer latent-factor mixture architectures. They do not directly address how skewness and heavy-tail variation should be propagated coherently through a hierarchy of multiple latent mixture transitions.

The proposed construction also connects directly to a broader methodological literature on flexible mixture modelling. Relevant precedents include likelihood-based multivariate skew-normal mixtures \citep{lin2009jmva}, common-factor-analyser mixtures for high-dimensional data \citep{wang2013jmva}, skew-normal extensions of mixtures of factor models \citep{lin2016jmva}, maximum-likelihood inference for multivariate $t$ mixtures \citep{wanglin2016jmva}, and more general scale/shape and normal mean--variance mixture constructions \citep{arellano2018jmva,naderi2019jmva}; see also the overview of skew distributions in model-based clustering by \citet{lee2022overview}. Collectively, these contributions provide important ingredients for accommodating skewness, heavy tails, factor-analytic dimension reduction, and normal mean--variance mixing, but they are principally single-layer formulations. The DStMM unifies these strands through a common positive scale propagated across an entire multilayer pathway, while retaining an exact observation-level component distribution.

To accommodate both features within such a deep hierarchy, we build on the generalised-hyperbolic skew-$t$ (GHST) normal mean--variance representation \citep{aas2006generalized,murray2014mixtures}. Let $\bm Y\in\mathbb R^p$ denote a generic observed random vector and let $\bm S=(S^{(1)},\ldots,S^{(h)})$ denote its complete latent component pathway, with a realised path written $\bm s=(s_1,\ldots,s_h)$. For pathway $\bm s$, write $\bm\mu_{\bm s}$ for the induced location vector, $\bm\alpha_{\bm s}$ for the induced skewness vector, $\bm\Sigma_{\bm s}$ for the positive-definite pathway scale matrix, and $\nu_{\bm s}$ for the degrees-of-freedom parameter; their recursive construction is given in Section~\ref{sec:path-collapse}. The proposed DStMM then induces
\begin{equation*}
\bm Y\mid (\bm S=\bm s,W)
\sim
\mathcal N_p\!\left(
\bm\mu_{\bm s}+W\bm\alpha_{\bm s},
W\bm\Sigma_{\bm s}
\right),
\qquad
W\sim\IG\!\left(\frac{\nu_{\bm s}}2,\frac{\nu_{\bm s}}2\right),
\end{equation*}
where $W>0$ is the pathway-specific mixing variable and $\IG(a,b)$ denotes the inverse-gamma distribution under the shape--scale parameterisation stated in Section~\ref{sec:ghst-kernel}. The central modelling choice is that the same $W$ is reused throughout all latent transitions on the selected pathway. It therefore controls both variance inflation and a directional mean displacement. Because the hierarchy remains conditionally Gaussian given $W$, the latent Gaussian states can be marginalised exactly, and each complete pathway reduces to a multivariate GHST component at the observation level. This pathway-level closure is the main structural distinction of the model: skewness and heavy-tail effects are propagated coherently through the entire deep hierarchy rather than being specified only within a single observed-to-factor transition. This closure is statistically useful beyond algebraic convenience. It supplies an exact observed-data density for every complete pathway, gives interpretable pathway-level location, skewness, and scale parameters, yields analytic posterior pathway probabilities, and provides a well-defined likelihood for comparing candidate architectures. The stochastic simulation used later is therefore required for local latent-state updates, not for evaluating the observed pathway density itself. The DStMM consequently forms a natural extension of the existing deep mixture hierarchy while connecting with the broader literature on skewed factor-analytic mixtures. Setting all local skewness vectors to zero gives the Student-$t$ RDMM, while increasing the degrees of freedom removes random scale mixing and yields the DGMM after absorbing the limiting skewness shift into the pathway location. With one latent layer, the construction reduces to a GHST factor-analytic mixture under the corresponding covariance restrictions, linking the proposed model to existing skew-$t$ factor-analytic approaches \citep{murray2014mixtures, lin2018robust, lee2021scale}. The principal distinction is therefore not merely the introduction of skewness into a factor model, but its propagation through multiple latent mixture layers using a common mixing variable that preserves an exact observation-level GHST representation for every complete pathway. These reductions make it possible to compare Gaussian, symmetric heavy-tailed, and skewed heavy-tailed deep mixtures within a common latent architecture.

The DStMM makes four main contributions. First, we derive backward recursions showing that each complete pathway induces an exact GHST component and hence an exact observed-data likelihood despite the presence of multiple latent mixture layers. Second, we formalise the principal nested-model reductions, state an explicit pathway-level separation assumption for the observed finite mixture, and distinguish this assumption from the local rotational and coordinate ambiguities of the deep parameterisation, while stating explicitly the counting convention used by the implementation. Third, we derive the conditional generalised inverse Gaussian law and associated moments and use them within a stochastic/Monte Carlo EM scheme based on analytic pathway responsibilities, conditional Gaussian simulation, and weighted regression updates; we also give the explicit parameter count used by the implementation-based BIC criterion for architecture selection. Fourth, the model retains the parsimonious layer-wise mixture structure of existing deep mixture models while allowing skewness to be activated only where needed. Here, parsimony refers to the layer-wise sharing of transition and mixing-weight parameters; the degrees-of-freedom structure can be specified pathway-wise or pooled according to the application. The empirical evaluation uses two complementary designs. In the simulation study and the handwritten-digit application, DGMM, RDMM, and DStMM are compared under a common prespecified architecture so that differences can be attributed directly to the component distribution. In the gas-sensor application, the architecture is selected by the corresponding implementation-based BIC criterion within each deep-model family, providing a complementary data-driven comparison of clustering performance and whether a non-trivial second layer is selected.

The remainder of the paper is organised as follows. Section~\ref{sec:model} develops the GHST pathway construction and its exact marginal representation. Section~\ref{sec:estimation} presents the stochastic/Monte Carlo EM algorithm. Section~\ref{sec:simulations} investigates the separate and joint effects of tail weight and asymmetry, together with robustness and parameter recovery. Section~\ref{sec:realdata} presents the handwritten-digit and gas-sensor applications, and Section~\ref{sec:discussion} summarises the empirical findings and discusses further methodological extensions. Additional simulation diagnostics are collected in the Appendix.

\section{A deep pathway construction with GHST marginals}
\label{sec:model}

Let $\bm Y_j\in\mathbb R^p$, $j=1,\ldots,n$, denote the random observations and write $\bm y_j$ for their realised values. We consider $h$ latent layers with
\begin{equation*}
p=r_0>r_1>\cdots>r_h\geq1,
\end{equation*}
and write $\bm z_j^{(0)}=\bm Y_j$ and $\bm z_j^{(l)}\in\mathbb R^{r_l}$ for the latent state at layer $l$. The term ``deep'' refers to this hierarchy of latent mixture transitions rather than to a deterministic neural-network architecture.

\subsection{GHST mean--variance mixture as the pathway kernel}
\label{sec:ghst-kernel}

We parameterise the inverse-gamma distribution by shape and scale:
\begin{equation*}
f_{\mathrm{IG}}(w;a,b)
=
\frac{b^a}{\Gamma(a)}w^{-a-1}\exp\!\left(-\frac{b}{w}\right),
\qquad w>0.
\end{equation*}
For a complete pathway $\bm s$, the mixing variable will satisfy
\begin{equation}
W_{j\bm s}\mid (\bm S_j=\bm s)
\sim
\IG\!\left(\frac{\nu_{\bm s}}2,\frac{\nu_{\bm s}}2\right).
\label{eq:path-scale-prior}
\end{equation}
The usual Student-$t$ scale mixture is recovered when the conditional mean does not depend on $W_{j\bm s}$. In the GHST case, the conditional distribution instead has the generic form
\begin{equation}
\bm X\mid (W=w)
\sim
\mathcal N_p(\bm\mu+w\bm\alpha,w\bm\Sigma).
\label{eq:generic-ghst-normal}
\end{equation}
When $\nu>2$, $\E(W)=\nu/(\nu-2)$ is finite. If $\bm\alpha\neq\bm0$ and a finite covariance matrix is required, $\nu>4$ is needed because
\[
\operatorname{Cov}(\bm X)
=
\E(W)\bm\Sigma+\operatorname{Var}(W)\bm\alpha\bm\alpha^\top.
\]
We use this distinction explicitly below rather than imposing $\nu>4$ in settings where only the density is required.

\subsection{Mixture paths and local transition parameters}
\label{sec:path-hierarchy}

Layer $l$ contains $K_l$ local components. Let $S_j^{(l)}\in\{1,\ldots,K_l\}$ be the component indicator for observation $j$ at that layer and assume the prior factorisation
\begin{equation*}
\Pr\{S_j^{(l)}=a\}=\pi_a^{(l)},
\qquad
\sum_{a=1}^{K_l}\pi_a^{(l)}=1.
\end{equation*}
A complete path is
\[
\bm S_j=(S_j^{(1)},\ldots,S_j^{(h)}),
\qquad
\bm s=(s_1,\ldots,s_h)\in\mathcal S,
\qquad
\mathcal S=\prod_{l=1}^h\{1,\ldots,K_l\}.
\]
Because the layer indicators are independent before observing $\bm y_j$,
\begin{equation*}
\pi_{\bm s}=\Pr(\bm S_j=\bm s)=\prod_{l=1}^h\pi_{s_l}^{(l)}.
\end{equation*}
Hence the mixing proportions require only $\sum_l(K_l-1)$ free parameters even though there are $\prod_lK_l$ complete paths.

For local component $a$ at layer $l$, introduce
\[
\bm\eta_a^{(l)}\in\mathbb R^{r_{l-1}},\qquad
\bm\Lambda_a^{(l)}\in\mathbb R^{r_{l-1}\times r_l},\qquad
\bm\Psi_a^{(l)}\in\mathbb R^{r_{l-1}\times r_{l-1}},\qquad
\bm\delta_a^{(l)}\in\mathbb R^{r_{l-1}}.
\]
The first three quantities are respectively a local intercept, loading matrix, and specific covariance matrix; $\bm\delta_a^{(l)}$ controls the scale-dependent mean shift. Conditional on a selected path and its mixing variable, the bottom of the hierarchy is standardised as
\begin{equation}
\bm z_j^{(h)}\mid (\bm S_j=\bm s,W_{j\bm s})
\sim
\mathcal N_{r_h}\!\left(\bm0,W_{j\bm s}\bm I_{r_h}\right),
\label{eq:deepest-new}
\end{equation}
and, for $l=h,h-1,\ldots,1$,
\begin{equation}
\bm z_j^{(l-1)}\mid
(\bm z_j^{(l)},\bm S_j=\bm s,W_{j\bm s})
\sim
\mathcal N_{r_{l-1}}\!\left(
\bm\eta_{s_l}^{(l)}+\bm\Lambda_{s_l}^{(l)}\bm z_j^{(l)}+W_{j\bm s}\bm\delta_{s_l}^{(l)},
\;W_{j\bm s}\bm\Psi_{s_l}^{(l)}
\right).
\label{eq:transition-new}
\end{equation}
The same $W_{j\bm s}$ appears in every conditional distribution along the path. It therefore controls both dispersion and directional displacement at the observation level after the latent variables are marginalised.

Because several local $\bm\delta_a^{(l)}$ vectors can contribute to the same aggregate pathway skewness, we use an active set $\mathcal A\subseteq\{1,\ldots,h\}$ to specify a parsimonious allocation of skewness across depth and impose
\begin{equation*}
\bm\delta_a^{(l)}=\bm0\quad\text{whenever }l\notin\mathcal A.
\end{equation*}
The baseline specification takes $\mathcal A=\{1\}$, which introduces component-specific skewness at the observation transition while retaining symmetric deeper transitions conditional on the common scale. This default gives a direct and interpretable decomposition of pathway skewness. The distributional recursions below remain valid for larger active sets, but richer allocations of skewness require additional identifiability constraints and are not covered by the baseline BIC dimension count in Section~\ref{sec:bic-complexity}.

\subsection{Collapsing a deep path}
\label{sec:path-collapse}

For fixed $\bm s$, define three sequences of induced parameters. Start at the deepest layer with
\begin{equation*}
\bm\mu_{\bm s}^{(h)}=\bm0,
\qquad
\bm\alpha_{\bm s}^{(h)}=\bm0,
\qquad
\bm\Sigma_{\bm s}^{(h)}=\bm I_{r_h},
\end{equation*}
and move upward through the hierarchy according to
\begin{align}
\bm\mu_{\bm s}^{(l-1)}
&=\bm\eta_{s_l}^{(l)}+\bm\Lambda_{s_l}^{(l)}\bm\mu_{\bm s}^{(l)},
\label{eq:mu-rec-new}\\
\bm\alpha_{\bm s}^{(l-1)}
&=\bm\delta_{s_l}^{(l)}+\bm\Lambda_{s_l}^{(l)}\bm\alpha_{\bm s}^{(l)},
\label{eq:alpha-rec-new}\\
\bm\Sigma_{\bm s}^{(l-1)}
&=\bm\Psi_{s_l}^{(l)}+\bm\Lambda_{s_l}^{(l)}\bm\Sigma_{\bm s}^{(l)}\bm\Lambda_{s_l}^{(l)\top}.
\label{eq:sigma-rec-new}
\end{align}
Write $\bm\mu_{\bm s}=\bm\mu_{\bm s}^{(0)}$, $\bm\alpha_{\bm s}=\bm\alpha_{\bm s}^{(0)}$ and $\bm\Sigma_{\bm s}=\bm\Sigma_{\bm s}^{(0)}$.

\begin{proposition}\label{prop:path-collapse}
Conditional on $\bm S_j=\bm s$ and $W_{j\bm s}=w$, every latent layer has a Gaussian marginal of the form
\begin{equation}
\bm z_j^{(l)}\mid (\bm S_j=\bm s,W_{j\bm s}=w)
\sim
\mathcal N_{r_l}\!\left(
\bm\mu_{\bm s}^{(l)}+w\bm\alpha_{\bm s}^{(l)},
\;w\bm\Sigma_{\bm s}^{(l)}
\right),
\qquad l=0,\ldots,h.
\label{eq:layer-marginal-prop}
\end{equation}
Consequently,
\begin{equation}
\bm Y_j\mid (\bm S_j=\bm s,W_{j\bm s})
\sim
\mathcal N_p\!\left(
\bm\mu_{\bm s}+W_{j\bm s}\bm\alpha_{\bm s},
\;W_{j\bm s}\bm\Sigma_{\bm s}
\right).
\label{eq:path-conditional-new}
\end{equation}
If every local specific covariance $\bm\Psi_{s_l}^{(l)}$ is positive definite, then every recursively induced covariance $\bm\Sigma_{\bm s}^{(l)}$ is positive definite.
\end{proposition}

\begin{proof}
The assertion is proved by backward induction on depth. At $l=h$, \eqref{eq:deepest-new} gives \eqref{eq:layer-marginal-prop} with $\bm\mu_{\bm s}^{(h)}=\bm\alpha_{\bm s}^{(h)}=\bm0$ and $\bm\Sigma_{\bm s}^{(h)}=\bm I_{r_h}$. Suppose the result holds at level $l$. Writing the transition in \eqref{eq:transition-new} as an affine transformation of $\bm z_j^{(l)}$ plus an independent Gaussian error shows that the marginal at level $l-1$ is Gaussian. Collecting the constant mean term, the coefficient of $w$, and the covariance coefficient gives exactly \eqref{eq:mu-rec-new}--\eqref{eq:sigma-rec-new}. Moreover, for any nonzero vector $\bm v$,
\[
\bm v^\top\bm\Sigma_{\bm s}^{(l-1)}\bm v
=
\bm v^\top\bm\Psi_{s_l}^{(l)}\bm v
+
(\bm\Lambda_{s_l}^{(l)\top}\bm v)^\top
\bm\Sigma_{\bm s}^{(l)}
(\bm\Lambda_{s_l}^{(l)\top}\bm v)>0,
\]
so positive definiteness is preserved. Iterating to $l=0$ proves the proposition.
\end{proof}

Proposition~\ref{prop:path-collapse} is the key closure property of the model. The local hierarchy may contain several latent mixture layers, yet from the observation level a fixed path is still a single normal mean--variance mixture. This allows likelihood calculations to be carried out with pathway-level GHST densities without numerically integrating over the Gaussian latent states.

\subsection{Exact pathway density and mixture likelihood}
\label{sec:path-density}

For later use, define
\begin{equation*}
d^2(\bm y;\bm\mu,\bm\Sigma)
=(\bm y-\bm\mu)^\top\bm\Sigma^{-1}(\bm y-\bm\mu),
\qquad
q(\bm\alpha;\bm\Sigma)=\bm\alpha^\top\bm\Sigma^{-1}\bm\alpha.
\end{equation*}
Under \eqref{eq:generic-ghst-normal} and the inverse-gamma law in \eqref{eq:path-scale-prior}, we write
\[
\bm Y\sim\GHSt_p(\bm\mu,\bm\Sigma,\bm\alpha,\nu).
\]
For $q(\bm\alpha;\bm\Sigma)>0$, its density is
\begin{align*}
g_{\mathrm{GHST},p}(\bm y;\bm\mu,\bm\Sigma,\bm\alpha,\nu)
={}&
\frac{\nu^{\nu/2}}
{(2\pi)^{p/2}|\bm\Sigma|^{1/2}\Gamma(\nu/2)2^{\nu/2-1}}
\exp\!\left\{(\bm y-\bm\mu)^\top\bm\Sigma^{-1}\bm\alpha\right\}\\
&\times
\left[
\frac{\nu+d^2(\bm y;\bm\mu,\bm\Sigma)}
{q(\bm\alpha;\bm\Sigma)}
\right]^{-(\nu+p)/4}
K_{(\nu+p)/2}\!\left(
\sqrt{q(\bm\alpha;\bm\Sigma)\{\nu+d^2(\bm y;\bm\mu,\bm\Sigma)\}}
\right),
\end{align*}
where $K_a(\cdot)$ denotes the modified Bessel function of the third kind. The continuous limit as $\bm\alpha\rightarrow\bm0$ is the multivariate Student-$t$ density.

Combining Proposition~\ref{prop:path-collapse} with \eqref{eq:path-scale-prior} gives
\begin{equation}
\bm Y_j\mid (\bm S_j=\bm s)
\sim
\GHSt_p(\bm\mu_{\bm s},\bm\Sigma_{\bm s},\bm\alpha_{\bm s},\nu_{\bm s}).
\label{eq:path-ghst-new}
\end{equation}
Let $\bm\Theta$ denote the collection of all free mixture weights, local transition parameters, active skewness parameters, and degrees-of-freedom parameters. The full observed density is therefore a finite mixture over complete paths,
\begin{equation*}
f(\bm y_j;\bm\Theta)
=
\sum_{\bm s\in\mathcal S}
\pi_{\bm s}
g_{\mathrm{GHST},p}(\bm y_j;\bm\mu_{\bm s},\bm\Sigma_{\bm s},\bm\alpha_{\bm s},\nu_{\bm s}),
\end{equation*}
with log-likelihood
\begin{equation}
\ell(\bm\Theta)
=
\sum_{j=1}^n
\log\!\left[
\sum_{\bm s\in\mathcal S}
\pi_{\bm s}
g_{\mathrm{GHST},p}(\bm y_j;\bm\mu_{\bm s},\bm\Sigma_{\bm s},\bm\alpha_{\bm s},\nu_{\bm s})
\right].
\label{eq:observed-loglik-new}
\end{equation}
Although the prior path probability factorises across layers, the posterior distribution of the layer indicators generally does not, because the pathway likelihood depends jointly on all indices in $\bm s$.

\subsection{Posterior path allocation and conditional mixing law}
\label{sec:posterior-model}

The posterior probability of path $\bm s$ is
\begin{equation}
\tau_{j\bm s}
=
\frac{
\pi_{\bm s}g_{\mathrm{GHST},p}(\bm y_j;\bm\mu_{\bm s},\bm\Sigma_{\bm s},\bm\alpha_{\bm s},\nu_{\bm s})
}{
\sum_{\bm r\in\mathcal S}
\pi_{\bm r}g_{\mathrm{GHST},p}(\bm y_j;\bm\mu_{\bm r},\bm\Sigma_{\bm r},\bm\alpha_{\bm r},\nu_{\bm r})
}.
\label{eq:tau-new}
\end{equation}
The marginal posterior probability for local component $a$ at layer $l$ follows by summing over paths containing that component:
\begin{equation*}
\tau_{ja}^{(l)}=\sum_{\bm s:s_l=a}\tau_{j\bm s}.
\end{equation*}

The scale update is conveniently expressed through the generalised inverse Gaussian distribution. We use
\begin{equation}
f_{\mathrm{GIG}}(w;\psi,\chi,\lambda)
=
\frac{(\psi/\chi)^{\lambda/2}}
{2K_\lambda(\sqrt{\psi\chi})}
 w^{\lambda-1}
\exp\!\left[-\frac12\left(\psi w+\frac{\chi}{w}\right)\right],
\qquad w>0.
\label{eq:gig-density-new}
\end{equation}
We write $W\sim\GIG(\psi,\chi,\lambda)$ for the distribution with density \eqref{eq:gig-density-new}. For each observation--path pair set
\begin{align}
D_{j\bm s}
&=(\bm y_j-\bm\mu_{\bm s})^\top\bm\Sigma_{\bm s}^{-1}(\bm y_j-\bm\mu_{\bm s}),
\label{eq:D-new}\\
A_{\bm s}
&=\bm\alpha_{\bm s}^\top\bm\Sigma_{\bm s}^{-1}\bm\alpha_{\bm s}.
\label{eq:A-new}
\end{align}
\begin{proposition}[Conditional mixing law]\label{prop:conditional-gig}
For a fixed observation--path pair, let $D_{j\bm s}$ and $A_{\bm s}$ be defined by \eqref{eq:D-new}--\eqref{eq:A-new}. If $A_{\bm s}>0$, then under the GIG parameterisation in \eqref{eq:gig-density-new},
\begin{equation*}
W_{j\bm s}\mid (\bm y_j,\bm S_j=\bm s)
\sim
\GIG\!\left(
A_{\bm s},\nu_{\bm s}+D_{j\bm s},-\frac{\nu_{\bm s}+p}{2}
\right).
\end{equation*}
Writing $\psi=A_{\bm s}$, $\chi=\nu_{\bm s}+D_{j\bm s}$ and $\lambda=-(\nu_{\bm s}+p)/2$, its conditional moments satisfy
\begin{equation}
\E(W_{j\bm s}^{q}\mid (\bm y_j,\bm S_j=\bm s))
=
\left(\frac{\chi}{\psi}\right)^{q/2}
\frac{K_{\lambda+q}(\sqrt{\psi\chi})}{K_{\lambda}(\sqrt{\psi\chi})}.
\label{eq:gig-moment-new}
\end{equation}
When $A_{\bm s}=0$, the continuous limit is
\begin{equation}
W_{j\bm s}\mid (\bm y_j,\bm S_j=\bm s)
\sim
\IG\!\left(
\frac{\nu_{\bm s}+p}{2},\frac{\nu_{\bm s}+D_{j\bm s}}{2}
\right),
\label{eq:W-symmetric-new}
\end{equation}
so that
\begin{equation}
\E(W_{j\bm s}^{-1}\mid (\bm y_j,\bm S_j=\bm s))
=
\frac{\nu_{\bm s}+p}{\nu_{\bm s}+D_{j\bm s}}.
\label{eq:t-weight-new}
\end{equation}
\end{proposition}

\begin{proof}
Expanding the quadratic form in \eqref{eq:path-conditional-new} gives a kernel proportional to
\[
w^{-(\nu_{\bm s}+p)/2-1}
\exp\!\left[
-\frac12
\left(
A_{\bm s}w
+\frac{\nu_{\bm s}+D_{j\bm s}}{w}
\right)
\right].
\]
after dropping the factor independent of $w$. This is the GIG kernel in \eqref{eq:gig-density-new} with the stated parameters. Equation~\eqref{eq:gig-moment-new} is the standard GIG moment identity under the same parameterisation. Setting $A_{\bm s}=0$ gives the inverse-gamma limit and \eqref{eq:t-weight-new} follows immediately.
\end{proof}

Thus one posterior variable governs two related effects: larger $W_{j\bm s}$ increases conditional dispersion and also shifts the conditional location along $\bm\alpha_{\bm s}$; in the symmetric limit, Proposition~\ref{prop:conditional-gig} recovers the familiar Student-$t$ down-weighting rule.

\subsection{Nested-model reductions}
\label{sec:nesting-identifiability}

\begin{corollary}[Principal reductions]\label{cor:nested-models}
Under the pathway construction in Proposition~\ref{prop:path-collapse}, the following reductions hold.
\begin{enumerate}
\item \emph{Symmetric heavy-tailed limit.} If $\bm\delta_a^{(l)}=\bm0$ for every $a$ and $l$, then $\bm\alpha_{\bm s}=\bm0$ for every complete path and the DStMM reduces exactly to the Student-$t$ RDMM of \citet{wu2026rdmm}.
\item \emph{Gaussian limit.} If $\nu_{\bm s}\rightarrow\infty$, then $W_{j\bm s}\rightarrow1$ in probability and
\[
\bm Y_j\mid (\bm S_j=\bm s)\;\Longrightarrow\;
\mathcal N_p(\bm\mu_{\bm s}+\bm\alpha_{\bm s},\bm\Sigma_{\bm s}).
\]
Equivalently, at the local level the limiting shift $\bm\delta_a^{(l)}$ can be absorbed into the corresponding intercept, giving the DGMM parameterisation.
\item \emph{Single-layer reduction.} When $h=1$, the observed mixture is a GHST mixture of factor analysers with
$\bm\mu_s=\bm\eta_s^{(1)}$, $\bm\alpha_s=\bm\delta_s^{(1)}$, and
$\bm\Sigma_s=\bm\Psi_s^{(1)}+\bm\Lambda_s^{(1)}\bm\Lambda_s^{(1)\top}$ under the baseline deepest-layer standardisation.
\end{enumerate}
\end{corollary}

\begin{proof}
Part 1 follows from the skewness recursion \eqref{eq:alpha-rec-new}. For Part 2, an inverse-gamma variable with equal shape and scale $\nu_{\bm s}/2$ converges in probability to one as $\nu_{\bm s}\rightarrow\infty$; Slutsky's theorem applied to \eqref{eq:path-conditional-new} gives the Gaussian limit. Part 3 follows by evaluating the recursions \eqref{eq:mu-rec-new}--\eqref{eq:sigma-rec-new} once from the standardised deepest layer.
\end{proof}

\subsection{Pathway-level separation and the role of the skewness active set}
\label{sec:identifiability}

Two distinct levels of identifiability are relevant. At the observation level, the model is a finite mixture of the pathway GHST distributions in \eqref{eq:path-ghst-new}. Rather than assert unrestricted global identifiability of the full GHST pathway family, we adopt the following \emph{pathway-separation assumption}: on the parameter region under consideration, if two finite positive-weight pathway mixtures define the same observed distribution, then their component weights and GHST component distributions agree up to a permutation of the pathway labels. This is the standard finite-mixture notion of identifiability modulo label switching. General criteria based on linear independence of the component family are given by \citet{yakowitz1968identifiability}, while closely related results for generalised-hyperbolic normal mean--variance mixtures are established by \citet{browne2015mixture}. These references provide theoretical context for the assumption; they are not invoked as a proof of unrestricted global identifiability for every boundary parameterisation of the present deep GHST model.

At the local level, factor-loading coordinates are not unique. For each layer $l$ and local component $a$, the M-step estimates $\bm\eta_a^{(l)}$, $\bm\Lambda_a^{(l)}$, and, on an active layer, $\bm\delta_a^{(l)}$ jointly by weighted regression. The baseline covariance update retains only the diagonal residual variances and bounds each diagonal entry below by the numerical floor $\psi_{\min}>0$. Hence positivity and diagonality of $\bm\Psi_a^{(l)}$ are imposed numerically, whereas the fitted loading matrices are not post-rotated to force
$(\bm\Lambda_a^{(l)})^\top(\bm\Psi_a^{(l)})^{-1}\bm\Lambda_a^{(l)}$
to be diagonal. For the implementation-based BIC count, the loading block is assigned the effective dimension
\begin{equation}
 d_{\Lambda,l}^{\mathrm{code}}
 =r_{l-1}r_l-\frac{r_l(r_l-1)}{2},
\label{eq:loading-ident-new}
\end{equation}
that is, the $r_{l-1}r_l$ loading coefficients minus the standard $r_l(r_l-1)/2$ rotational non-identifiability correction adopted in the implementation-based parameter count. Equation~\eqref{eq:loading-ident-new} is therefore a parameter-counting convention rather than a constraint imposed on the stored loading matrices.

The hierarchy also admits intermediate latent-coordinate reparameterisations. For $l=1,\ldots,h-1$, let $\bm D_l$ be any positive diagonal $r_l\times r_l$ matrix and let $\bm c_l\in\mathbb R^{r_l}$. Replacing
\begin{align*}
\widetilde{\bm\eta}_a^{(l)}
&= \bm\eta_a^{(l)}
 - \bm\Lambda_a^{(l)}\bm D_l^{-1}\bm c_l,
&
\widetilde{\bm\Lambda}_a^{(l)}
&= \bm\Lambda_a^{(l)}\bm D_l^{-1},
\\[3pt]
\widetilde{\bm\eta}_b^{(l+1)}
&= \bm D_l\bm\eta_b^{(l+1)}+\bm c_l,
&
\widetilde{\bm\Lambda}_b^{(l+1)}
&= \bm D_l\bm\Lambda_b^{(l+1)},
\\[-1pt]
\widetilde{\bm\Psi}_b^{(l+1)}
&= \bm D_l\bm\Psi_b^{(l+1)}\bm D_l,
&
\widetilde{\bm\delta}_b^{(l+1)}
&= \bm D_l\bm\delta_b^{(l+1)}
\end{align*}
leaves the induced observed distribution unchanged when the adjacent latent coordinates are transformed coherently. The fitting implementation leaves this intermediate-coordinate gauge unfixed: it does not impose a reference-component normalisation such as $\bm\eta_1^{(l+1)}=\bm0$ and $\bm\Psi_1^{(l+1)}=\bm I$, and its BIC count does not subtract additional gauge dimensions. Accordingly, the model-selection criterion reported below uses the software parameter count rather than a separately gauge-quotiented model dimension.

Skewness requires a separate allocation restriction. From \eqref{eq:alpha-rec-new}, several nonzero local vectors can contribute to the same aggregate $\bm\alpha_{\bm s}$. Our baseline choice $\mathcal A=\{1\}$ assigns directional asymmetry to the observation transition and leaves deeper transitions conditionally symmetric. This restriction is implemented directly by estimating $\bm\delta_a^{(l)}$ only when $l\in\mathcal A$ and setting the remaining local skewness vectors to zero. Richer active sets are available in the implementation, but they can introduce additional allocation ambiguities and therefore require care when interpreting the corresponding parameter count.

\section{Estimation by stochastic/Monte Carlo EM}
\label{sec:estimation}

The normal mean--variance mixture representation leads to a convenient augmented-data scheme. We retain complete-path probabilities as soft weights, conditionally simulate the positive mixing variables and Gaussian latent states, and then update local parameters from weighted Gaussian regressions. With one simulated draw per observation--path pair, the procedure has the character of a stochastic EM algorithm; using several draws produces a Monte Carlo EM approximation with reduced simulation noise \citep{nielsen2000stochastic,levine2001mcem}.

More explicitly, write
\[
\bm\Theta=
\left\{
\pi_a^{(l)},
\bm\eta_a^{(l)},
\bm\Lambda_a^{(l)},
\bm\Psi_a^{(l)},
\bm\delta_a^{(l)},
\nu_{\bm s}
\right\},
\]
where $\bm\delta_a^{(l)}=\bm0$ for inactive skewness layers.

\subsection{Augmented likelihood}
\label{sec:augmented-likelihood}

Let $u_{j\bm s}=\mathbb I(\bm S_j=\bm s)$, where $\mathbb I(\cdot)$ is the indicator function, and let $\phi_d(\cdot;\bm m,\bm V)$ denote the $d$-variate Gaussian density with mean $\bm m$ and covariance matrix $\bm V$. Conditional on $u_{j\bm s}=1$, the missing quantities are $W_{j\bm s}$ and $\bm z_j^{(1)},\ldots,\bm z_j^{(h)}$. Apart from constants independent of $\bm\Theta$, the complete-data log-likelihood is
\begin{align}
\ell_c(\bm\Theta)
=
\sum_{j=1}^n\sum_{\bm s\in\mathcal S}u_{j\bm s}
\Bigg\{
&\sum_{l=1}^h\log\pi_{s_l}^{(l)}
+
\log f_{\mathrm{IG}}\!\left(W_{j\bm s};\frac{\nu_{\bm s}}2,\frac{\nu_{\bm s}}2\right)
\nonumber\\
&+
\log\phi_{r_h}\!\left(\bm z_j^{(h)};\bm0,W_{j\bm s}\bm I_{r_h}\right)
\nonumber\\
&+
\sum_{l=1}^h
\log\phi_{r_{l-1}}\!\left(
\bm z_j^{(l-1)};
\bm\eta_{s_l}^{(l)}+\bm\Lambda_{s_l}^{(l)}\bm z_j^{(l)}+W_{j\bm s}\bm\delta_{s_l}^{(l)},
W_{j\bm s}\bm\Psi_{s_l}^{(l)}
\right)
\Bigg\}.
\label{eq:complete-loglik-new}
\end{align}
The deepest Gaussian term carries no free mean or scale-matrix parameters under the standardisation in \eqref{eq:deepest-new}.

\subsection{Path probabilities and scale calculations}
\label{sec:path-scale-estimation}

Suppose $\bm\Theta^{(t)}$ is available. The first calculation at iteration $t$ is the set of pathway probabilities $\tau_{j\bm s}^{(t)}$ from \eqref{eq:tau-new}. In numerical work, we evaluate
\[
a_{j\bm s}^{(t)}
=
\log\pi_{\bm s}^{(t)}
+
\log g_{\mathrm{GHST},p}\!\left(
\bm y_j;
\bm\mu_{\bm s}^{(t)},
\bm\Sigma_{\bm s}^{(t)},
\bm\alpha_{\bm s}^{(t)},
\nu_{\bm s}^{(t)}
\right)
\]
and normalise with a log-sum-exp operation. When $A_{\bm s}^{(t)}$ is numerically negligible, the Student-$t$ limiting density is used directly instead of evaluating the Bessel-function form at a near-degenerate skewness value.

For each $j$ and $\bm s$, form $D_{j\bm s}^{(t)}$ and $A_{\bm s}^{(t)}$ from \eqref{eq:D-new}--\eqref{eq:A-new}. The positive mixing variable may then be sampled as
\begin{equation}
W_{j\bm s,m}^{(t)}
\sim
\GIG\!\left(
A_{\bm s}^{(t)},
\nu_{\bm s}^{(t)}+D_{j\bm s}^{(t)},
-\frac{\nu_{\bm s}^{(t)}+p}{2}
\right),
\qquad m=1,\ldots,M.
\label{eq:W-draw-new}
\end{equation}
For symmetric paths, the inverse-gamma law in \eqref{eq:W-symmetric-new} is used instead.

Analytic GIG moments are useful for stabilising updates that depend only on $W$. Equation~\eqref{eq:gig-moment-new} in Proposition~\ref{prop:conditional-gig} supplies $\E(W^q)$ directly; in particular, $q=-1$ gives the precision weight $\E(W^{-1})$. The logarithmic moment can be written as
\begin{equation}
\E(\log W)
=
\frac12\log\!\left(\frac{\chi}{\psi}\right)
+
\frac{\partial}{\partial\lambda}\log K_\lambda(\sqrt{\psi\chi}),
\label{eq:gig-log-moment}
\end{equation}
which is preferable to a Monte Carlo approximation when a stable numerical derivative of the Bessel function is available.

\subsection{Conditional simulation of the latent hierarchy}
\label{sec:latent-simulation}

The latent states are generated from the observation layer toward the deepest layer. For fixed $\bm s$ and $W_{j\bm s}=w$, Proposition~\ref{prop:path-collapse} gives the prior marginal
\[
\bm z_j^{(l)}\mid (\bm S_j=\bm s,W_{j\bm s}=w)
\sim
\mathcal N_{r_l}\!\left(
\bm\mu_{\bm s}^{(l)}+w\bm\alpha_{\bm s}^{(l)},
w\bm\Sigma_{\bm s}^{(l)}
\right).
\]
Combining this distribution with the transition for $\bm z_j^{(l-1)}$ gives
\begin{equation}
\bm z_j^{(l)}\mid
(\bm z_j^{(l-1)},\bm S_j=\bm s,W_{j\bm s}=w)
\sim
\mathcal N_{r_l}\!\left(
\bm\rho_{j\bm s}^{(l)}(w),
w\bm\Omega_{\bm s}^{(l)}
\right),
\label{eq:latent-cond-new}
\end{equation}
where
\begin{equation*}
\bm\Omega_{\bm s}^{(l)}
=
\left[
(\bm\Sigma_{\bm s}^{(l)})^{-1}
+
\bm\Lambda_{s_l}^{(l)\top}(\bm\Psi_{s_l}^{(l)})^{-1}\bm\Lambda_{s_l}^{(l)}
\right]^{-1},
\end{equation*}
and
\begin{align*}
\bm\rho_{j\bm s}^{(l)}(w)
=\bm\Omega_{\bm s}^{(l)}\Big[&
(\bm\Sigma_{\bm s}^{(l)})^{-1}
\{\bm\mu_{\bm s}^{(l)}+w\bm\alpha_{\bm s}^{(l)}\}\\
&+
\bm\Lambda_{s_l}^{(l)\top}(\bm\Psi_{s_l}^{(l)})^{-1}
\{\bm z_j^{(l-1)}-\bm\eta_{s_l}^{(l)}-w\bm\delta_{s_l}^{(l)}\}
\Big].
\end{align*}
For replicate $m$, set $\bm z_{j\bm s,m}^{(0,t)}=\bm y_j$ and use the same draw $W_{j\bm s,m}^{(t)}$ while sampling successively from $l=1$ to $h$. Reusing the scale draw is essential: replacing it by independent layer-specific draws would define a different model and would destroy the exact pathway GHST marginal in \eqref{eq:path-ghst-new}.

\subsection{Regression-form updates for local parameters}
\label{sec:regression-updates}

Conditional on the augmented variables, each transition is a multivariate Gaussian regression with heteroscedastic weight $W^{-1}$. For $l\in\mathcal A$, define
\begin{equation*}
\widetilde{\bm x}_{j\bm s,m}^{(l,t)}
=
\begin{pmatrix}
1\\
\bm z_{j\bm s,m}^{(l,t)}\\
W_{j\bm s,m}^{(t)}
\end{pmatrix},
\qquad
\bm B_a^{(l)}=
\begin{pmatrix}
\bm\eta_a^{(l)} & \bm\Lambda_a^{(l)} & \bm\delta_a^{(l)}
\end{pmatrix}.
\end{equation*}
For $l\notin\mathcal A$, omit the final coordinate and the skewness column. Define the Monte Carlo cross-products
\begin{align}
\widehat{\bm R}_{j\bm s}^{(l,t)}
&=
\frac1M\sum_{m=1}^M
\frac{1}{W_{j\bm s,m}^{(t)}}
\bm z_{j\bm s,m}^{(l-1,t)}
\widetilde{\bm x}_{j\bm s,m}^{(l,t)\top},
\label{eq:Rjs-new}\\
\widehat{\bm Q}_{j\bm s}^{(l,t)}
&=
\frac1M\sum_{m=1}^M
\frac{1}{W_{j\bm s,m}^{(t)}}
\widetilde{\bm x}_{j\bm s,m}^{(l,t)}
\widetilde{\bm x}_{j\bm s,m}^{(l,t)\top}.
\label{eq:Qjs-new}
\end{align}
Pooling over all paths that use local component $a$ gives
\begin{align*}
\bm R_a^{(l,t)}
&=\sum_{j=1}^n\sum_{\bm s:s_l=a}\tau_{j\bm s}^{(t)}\widehat{\bm R}_{j\bm s}^{(l,t)},\\
\bm Q_a^{(l,t)}
&=\sum_{j=1}^n\sum_{\bm s:s_l=a}\tau_{j\bm s}^{(t)}\widehat{\bm Q}_{j\bm s}^{(l,t)}.
\end{align*}
The coefficient update is therefore
\begin{equation}
\bm B_a^{(l,t+1)}
=
\bm R_a^{(l,t)}\{\bm Q_a^{(l,t)}\}^{-1}.
\label{eq:B-update-new}
\end{equation}
This single regression step updates the local intercept and loading matrix, and, on active layers, the local skewness vector. For additional numerical stabilisation, a generalised inverse or a small ridge correction can be used for $\bm Q_a^{(l,t)}$.

The layer mixing proportions are updated from marginalised pathway mass:
\begin{equation}
\pi_a^{(l,t+1)}
=
\frac1n\sum_{j=1}^n\sum_{\bm s:s_l=a}\tau_{j\bm s}^{(t)}.
\label{eq:pi-update-new}
\end{equation}
Let
\[
N_a^{(l,t)}=\sum_{j=1}^n\sum_{\bm s:s_l=a}\tau_{j\bm s}^{(t)}.
\]
With
\[
\bm e_{j\bm s,m}^{(l,t+1)}
=
\bm z_{j\bm s,m}^{(l-1,t)}
-
\bm B_a^{(l,t+1)}\widetilde{\bm x}_{j\bm s,m}^{(l,t)},
\]
the unrestricted specific covariance update is
\begin{equation}
\bm\Psi_a^{(l,t+1)}
=
\frac{1}{N_a^{(l,t)}}
\sum_{j=1}^n\sum_{\bm s:s_l=a}\tau_{j\bm s}^{(t)}
\frac1M\sum_{m=1}^M
\frac{\bm e_{j\bm s,m}^{(l,t+1)}\bm e_{j\bm s,m}^{(l,t+1)\top}}
{W_{j\bm s,m}^{(t)}}.
\label{eq:Psi-update-new}
\end{equation}
The baseline diagonal model replaces the right-hand side by its diagonal part and imposes a small positive lower bound on the diagonal entries.

\subsection{Updating the degrees of freedom}
\label{sec:nu-update}

Let $N_{\bm s}^{(t)}=\sum_j\tau_{j\bm s}^{(t)}$. Differentiating the inverse-gamma contribution in \eqref{eq:complete-loglik-new} yields the scalar equation
\begin{equation}
\log\!\left(\frac{\nu_{\bm s}}2\right)
-\psi\!\left(\frac{\nu_{\bm s}}2\right)
+1
-
\frac{1}{N_{\bm s}^{(t)}}
\sum_{j=1}^n\tau_{j\bm s}^{(t)}
\left\{
\widehat{\ell}_{j\bm s}^{(t)}+\widehat{b}_{j\bm s}^{(t)}
\right\}
=0,
\label{eq:nu-update-new}
\end{equation}
where $\psi(\cdot)$ denotes the digamma function and
\[
\widehat{\ell}_{j\bm s}^{(t)}
=\E\{\log W_{j\bm s}\mid (\bm y_j,\bm S_j=\bm s;\bm\Theta^{(t)})\},
\qquad
\widehat{b}_{j\bm s}^{(t)}
=\E\{W_{j\bm s}^{-1}\mid (\bm y_j,\bm S_j=\bm s;\bm\Theta^{(t)})\}.
\]
The expectations can be evaluated from \eqref{eq:gig-moment-new}--\eqref{eq:gig-log-moment}, thereby avoiding unnecessary Monte Carlo noise in the degrees-of-freedom step. We solve \eqref{eq:nu-update-new} by a safeguarded one-dimensional root finder on $\nu_{\min}\leq\nu_{\bm s}\leq\nu_{\max}$. The lower bound may be set just above two when only the density is of interest, or above four when finite covariance under nonzero skewness is required. When additional pooling across pathways is desirable, a common $\nu$, a first-layer-specific $\nu_{s_1}$, or another pooled structure can be used.

\subsection{Fitting workflow}
\label{sec:fitting-workflow}

Algorithm~\ref{alg:dstmm} summarises one iteration. A fitted DGMM or RDMM provides a useful starting point; in particular, an RDMM warm start makes it possible to introduce skewness gradually from a well-fitted symmetric model.

\begin{algorithm}[H]
\caption{One stochastic/Monte Carlo EM iteration for the DStMM}
\label{alg:dstmm}
\begin{algorithmic}[1]
\Require Current parameters $\bm\Theta^{(t)}$ and Monte Carlo size $M$.
\State Collapse every path using \eqref{eq:mu-rec-new}--\eqref{eq:sigma-rec-new} to obtain $\bm\mu_{\bm s}^{(t)}$, $\bm\alpha_{\bm s}^{(t)}$ and $\bm\Sigma_{\bm s}^{(t)}$.
\State Evaluate $\tau_{j\bm s}^{(t)}$ from the observed GHST mixture on the log scale.
\For{each observation $j$ and path $\bm s$}
    \State Draw $W_{j\bm s,1:M}^{(t)}$ from \eqref{eq:W-draw-new}.
    \State For every scale draw, sample $\bm z^{(1)},\ldots,\bm z^{(h)}$ sequentially from \eqref{eq:latent-cond-new} using that same $W$ throughout the path.
\EndFor
\State Form the weighted cross-products in \eqref{eq:Rjs-new}--\eqref{eq:Qjs-new}.
\State Update $\bm\pi$, $\bm\eta$, $\bm\Lambda$, $\bm\delta$ and $\bm\Psi$ using \eqref{eq:B-update-new}, \eqref{eq:pi-update-new} and \eqref{eq:Psi-update-new}.
\State Update the degrees of freedom by solving \eqref{eq:nu-update-new}.
\State Impose the diagonal specific-covariance update and its positive lower bound; evaluate \eqref{eq:observed-loglik-new}.
\end{algorithmic}
\end{algorithm}

When $M=1$, the procedure is a stochastic-EM approximation rather than a deterministic EM algorithm, so monotone increase of the observed log-likelihood at every iteration is not guaranteed. We therefore monitor a moving-window summary of \eqref{eq:observed-loglik-new} and retain covariance lower bounds after every update; the simulation settings used in this paper are reported explicitly in Section~\ref{sec:sim-common}. A gradually increasing Monte Carlo size or $M>1$ can be used when a smoother Monte Carlo approximation is required. Multiple starts or warm starts from the nested RDMM provide broader exploration of the multimodal deep-mixture likelihood, and the retained fit should be chosen by the largest observed-data log-likelihood rather than by external class labels. In larger architectures, paths with persistently small posterior mass can be pruned, and the active-set formulation provides a direct way to match skewness complexity to the available information.

\subsection{Computational scaling and likelihood-based model selection}
\label{sec:bic-complexity}

Let $|\mathcal S|=\prod_{l=1}^hK_l$ denote the number of complete paths. With direct dense pathway calculations, forming and factorising the collapsed $p\times p$ covariance matrices costs $O(|\mathcal S|p^3)$ per parameter update when those factors must be refreshed, after which likelihood evaluation over all observations costs $O(n|\mathcal S|p^2)$. Conditional latent-state simulation adds approximately
$O\{Mn|\mathcal S|\sum_{l=1}^h r_{l-1}r_l\}$ matrix--vector work once the path/layer conditional covariance factors have been cached. These expressions explain the principal computational limitation of increasing depth: even though local parameters and prior weights are shared layer-wise, exact posterior allocation enumerates the complete paths. Pathwise density evaluations and conditional simulations are independent conditional on the current parameters and can therefore be parallelised.

For a candidate architecture $\mathcal M$, the implementation-based BIC criterion used in this paper is
\begin{equation*}
\operatorname{BIC}(\mathcal M)
=-2\ell(\widehat{\bm\Theta}_{\mathcal M})+d_{\mathrm{code}}(\mathcal M)\log n,
\end{equation*}
so smaller values are preferred. Here $d_{\mathrm{code}}(\mathcal M)$ is the parameter count returned by the fitting implementation. We use this criterion consistently for within-family architecture comparison; it should be interpreted as the software BIC convention rather than as a claim that $d_{\mathrm{code}}(\mathcal M)$ is a separately gauge-quotiented model-manifold dimension. For DStMM, write $r_0=p$. At layer $l$, the code counts $K_l-1$ free mixing proportions, $K_l r_{l-1}$ intercept coefficients, $K_l r_{l-1}$ diagonal specific-variance parameters, and the loading contribution in \eqref{eq:loading-ident-new}. If layer $l$ belongs to the skewness active set, it also counts $K_l r_{l-1}$ skewness coefficients. Hence the implemented DStMM count is
\begin{equation}
\begin{split}
 d_{\mathrm{code}}(\mathcal M)
 ={}&\sum_{l=1}^{h}
 \Bigg[
 (K_l-1)
 +K_l\left\{r_{l-1}r_l+2r_{l-1}
 -\frac{r_l(r_l-1)}{2}\right\}
 \\[2pt]
 &\hspace{43mm}
 +\mathbb I(l\in\mathcal A)K_l r_{l-1}
 \Bigg]
 +d_\nu .
\end{split}
\label{eq:bic-df-new}
\end{equation}
The degrees-of-freedom contribution is exactly the one selected by the software option: $d_\nu=1$ for a common $\nu$, $d_\nu=K_1$ for a first-layer-specific $\nu$, and $d_\nu=\prod_{l=1}^hK_l$ for pathway-specific degrees of freedom. In the baseline empirical DStMM fits, $\mathcal A=\{1\}$. No additional subtraction of $2\sum_{l=1}^{h-1}r_l$ is made in the current code.

The term $r_l(r_l-1)/2$ in \eqref{eq:bic-df-new} is best interpreted as the standard rotational non-identifiability correction adopted in the implementation-based parameter count for each local loading block. It should not be read as evidence that the stored fitted matrices have been explicitly rotated to satisfy a condition such as $(\bm\Lambda_a^{(l)})^\top(\bm\Psi_a^{(l)})^{-1}\bm\Lambda_a^{(l)}$ being diagonal. The M-step stores the weighted-regression loading estimate directly, while the covariance update explicitly retains only the diagonal residual variances. Thus the likelihood calculation uses the fitted parameter matrices as stored, whereas the BIC uses the effective count in \eqref{eq:bic-df-new}.

For the numerical model-selection results, we use the BIC returned by the corresponding fitting implementation for each model family. Equation~\eqref{eq:bic-df-new} records the DStMM convention explicitly and is used consistently across all DStMM candidate architectures.

\FloatBarrier
\section{Numerical simulations}
\label{sec:simulations}

The experiments are organised around a nesting question rather than around a direct repetition of the RDMM robustness study. The first experiment starts exactly from a symmetric Student-$t$ deep mixture and gradually introduces directional asymmetry. The second crosses several tail-weight levels with several skewness levels. This design distinguishes three regimes: a Gaussian model is sufficient, a symmetric heavy-tailed model is sufficient, or a genuinely skewed heavy-tailed pathway model is needed.

\subsection{Shared architecture and controlled skewness mechanism}
\label{sec:sim-common}

All data sets are generated from a two-layer hierarchy with
\begin{equation*}
p=20,
\qquad
(r_1,r_2)=(5,2),
\qquad
(K_1,K_2)=(2,2),
\end{equation*}
so there are four complete pathways. The layer mixing proportions are balanced,
\[
\bm\pi^{(1)}=\bm\pi^{(2)}=(0.5,0.5)^\top.
\]
To keep the symmetric baseline comparable with the RDMM study of \citet{wu2026rdmm}, we retain its location, loading, and specific-covariance configuration. Let $\bm e_q^{(d)}$ denote the $q$th standard basis vector in $\mathbb R^d$. The unadjusted first-layer locations are
\[
\bm\eta_{1,0}^{(1)}=-2.5\sum_{q=1}^3\bm e_q^{(20)},
\qquad
\bm\eta_{2,0}^{(1)}=2.5\sum_{q=1}^3\bm e_q^{(20)},
\]
and the second-layer locations are
\[
\bm\eta_1^{(2)}=1.25\bm e_1^{(5)},
\qquad
\bm\eta_2^{(2)}=-1.25\bm e_1^{(5)}.
\]
Let
\[
\bm a=(1.0,0.9,1.1,0.8)^\top,
\qquad
\bm A=\bm I_5\otimes\bm a,
\]
where $\otimes$ denotes the Kronecker product, and
\[
\bm D_2=\operatorname{diag}(1.10,0.90,1.05,0.95,1.00),
\]
with
\[
\bm\Lambda_1^{(1)}=\bm A,
\qquad
\bm\Lambda_2^{(1)}=\bm A\bm D_2,
\]
and
\[
\bm\Lambda_1^{(2)}=
\begin{pmatrix}
1.0&0.0\\
0.8&0.2\\
0.0&1.0\\
0.2&0.8\\
0.6&-0.6
\end{pmatrix},
\qquad
\bm\Lambda_2^{(2)}=
\begin{pmatrix}
0.9&0.1\\
0.7&-0.2\\
0.1&0.9\\
-0.2&0.7\\
0.5&0.5
\end{pmatrix}.
\]
The specific covariance matrices are
\[
\bm\Psi_a^{(1)}=0.5\bm I_{20},\qquad a=1,2,
\qquad
\bm\Psi_b^{(2)}=0.3\bm I_5,\qquad b=1,2.
\]

The data-generating model uses the baseline active set $\mathcal A=\{1\}$. We introduce skewness along the unit vector
\begin{equation*}
\bm v=\frac12(1,1,1,1,0,\ldots,0)^\top\in\mathbb R^{20}
\end{equation*}
and set
\begin{equation*}
\bm\delta_1^{(1)}=\bm\delta_2^{(1)}=\kappa\bm v,
\qquad
\bm\delta_a^{(2)}=\bm0,
\end{equation*}
where $\kappa\geq0$ controls the magnitude of directional skewness. Using the same direction for both first-layer components changes their shapes without intentionally widening the separation between their centres. Since $\E(W)=\nu/(\nu-2)$, the first-layer locations are centred according to
\begin{equation*}
\bm\eta_a^{(1)}(\nu,\kappa)
=
\bm\eta_{a,0}^{(1)}
-
\frac{\nu}{\nu-2}\kappa\bm v,
\qquad a=1,2.
\end{equation*}
Consequently, changing $\kappa$ alters asymmetry while leaving the unconditional first-layer centre fixed. Each generated observation first receives a complete path, and then one shared variable
\[
W_j\sim\IG\!\left(\frac\nu2,\frac\nu2\right),
\]
which is reused in both latent transitions and in the observation equation.

Although the general formulation permits pathway-specific degrees of freedom $\nu_{\bm s}$, the simulation study uses a common $\nu$ across pathways to isolate the effects of tail weight and directional asymmetry and to keep the comparison across models focused on component shape. We fit DGMM, RDMM, and DStMM under the same known architecture. In the full runs, the maximum number of iterations is 300, the numerical tolerance is $10^{-4}$, the minimum iteration count is 100, and the moving convergence window is 20. The DStMM uses Monte Carlo size $M=1$; RDMM and DStMM use a diagonal covariance floor of $10^{-6}$, and their common degrees of freedom are constrained to $[2.01,200]$. Fits are performed in the order DGMM, RDMM, and DStMM, with the DStMM warm-started from the corresponding RDMM fit whenever available. Because the data-generating architecture is fixed and known, architecture selection is not mixed into the comparison of component shapes.

\subsection{Numerical simulation 1}
\label{sec:sim1}

The first experiment fixes $\nu=6$ and varies
\begin{equation*}
\kappa\in\{0,0.5,1.0,1.5,2.0\}.
\end{equation*}
For every value of $\kappa$, the main design uses $n=1000$ observations and $R=1000$ independent Monte Carlo replications. Two robustness checks repeat the full grid with $n=500$ and with moderately unbalanced layer proportions $\bm\pi^{(1)}=\bm\pi^{(2)}=(0.7,0.3)^\top$, respectively. The reported summaries are the first-layer and complete-pathway adjusted Rand index (ARI), the corresponding misclassification rates (MR), recovery of the common degrees of freedom, and recovery of the first-layer skewness vectors. The latter is measured by
\begin{equation*}
\operatorname{RMSE}_{\delta}
=
\left[
\frac{1}{2p}
\sum_{a=1}^2
\|\widehat{\bm\delta}_a^{(1)}-\bm\delta_a^{(1)}\|^2
\right]^{1/2}.
\end{equation*}

\begin{table}[!ht]
\centering
\caption{Experiment 1 complete-pathway ARI over 1000 Monte Carlo replications. Entries are mean (standard deviation).}
\label{tab:exp1-path-ari-main}
\setlength{\tabcolsep}{3.5pt}
\begin{tabular}{lccccc}
\toprule
& \multicolumn{5}{c}{$\kappa$} \\
\cmidrule(lr){2-6}
Model & 0 & 0.5 & 1.0 & 1.5 & 2.0 \\
\midrule
DGMM & 0.489 (0.102) & 0.456 (0.101) & 0.363 (0.123) & 0.233 (0.140) & 0.132 (0.126) \\
RDMM & 0.767 (0.061) & 0.748 (0.074) & 0.708 (0.082) & 0.656 (0.093) & 0.601 (0.097) \\
DStMM & 0.770 (0.054) & 0.761 (0.067) & 0.742 (0.079) & 0.722 (0.091) & 0.695 (0.102) \\
\bottomrule
\end{tabular}
\end{table}

\begin{figure}[!ht]
\centering
\begin{subfigure}[t]{0.65\linewidth}
    \centering
    \includegraphics[width=\linewidth]{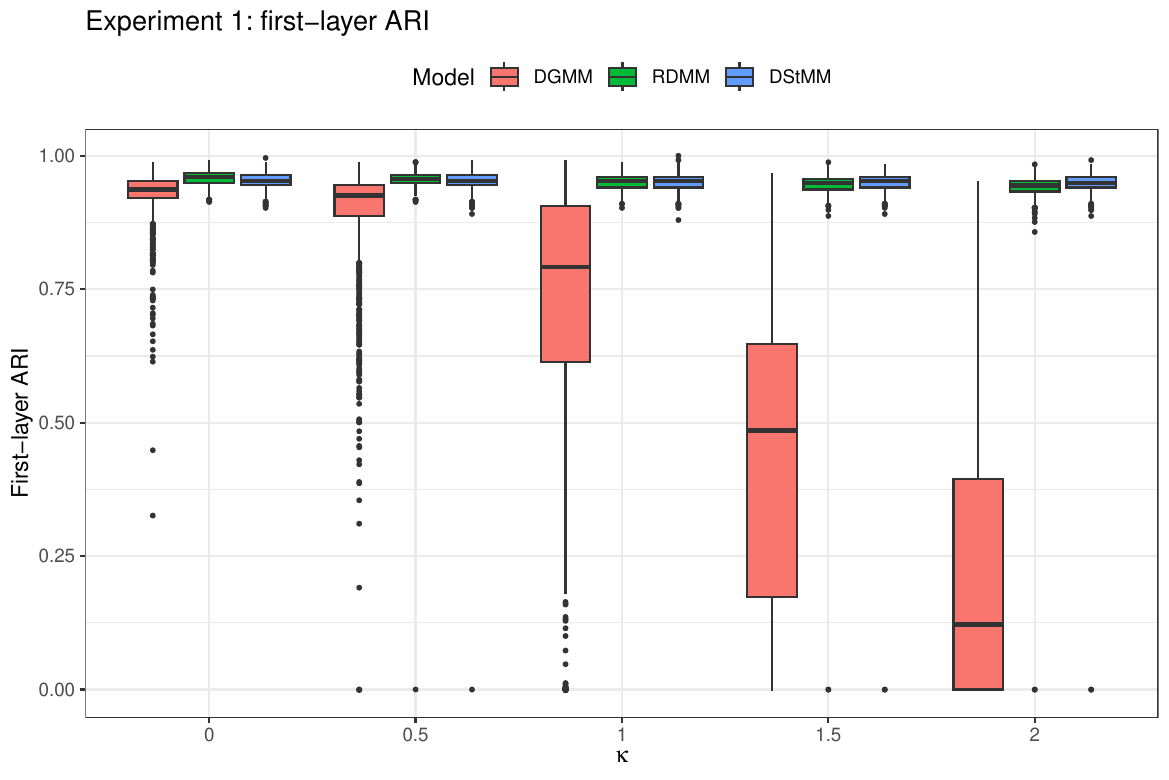}
    \caption{First-layer ARI.}
    \label{fig:exp1-l1-ari}
\end{subfigure}

\medskip
\begin{subfigure}[t]{0.65\linewidth}
    \centering
    \includegraphics[width=\linewidth]{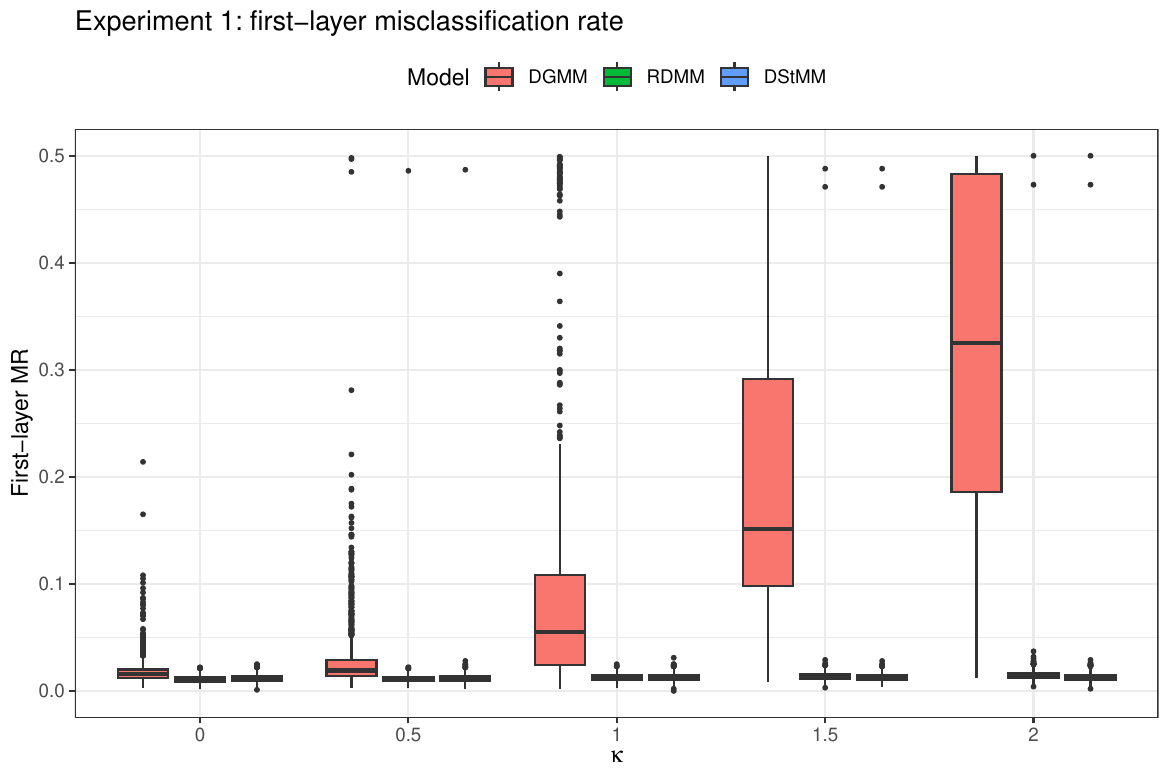}
    \caption{First-layer MR.}
    \label{fig:exp1-l1-mr}
\end{subfigure}
\caption{Experiment 1 first-layer clustering performance over 1000 Monte Carlo replications as directional skewness increases.}
\label{fig:exp1-first-layer-composite}
\end{figure}

\begin{figure}[!ht]
\centering
\begin{subfigure}[t]{0.65\linewidth}
    \centering
    \includegraphics[width=\linewidth]{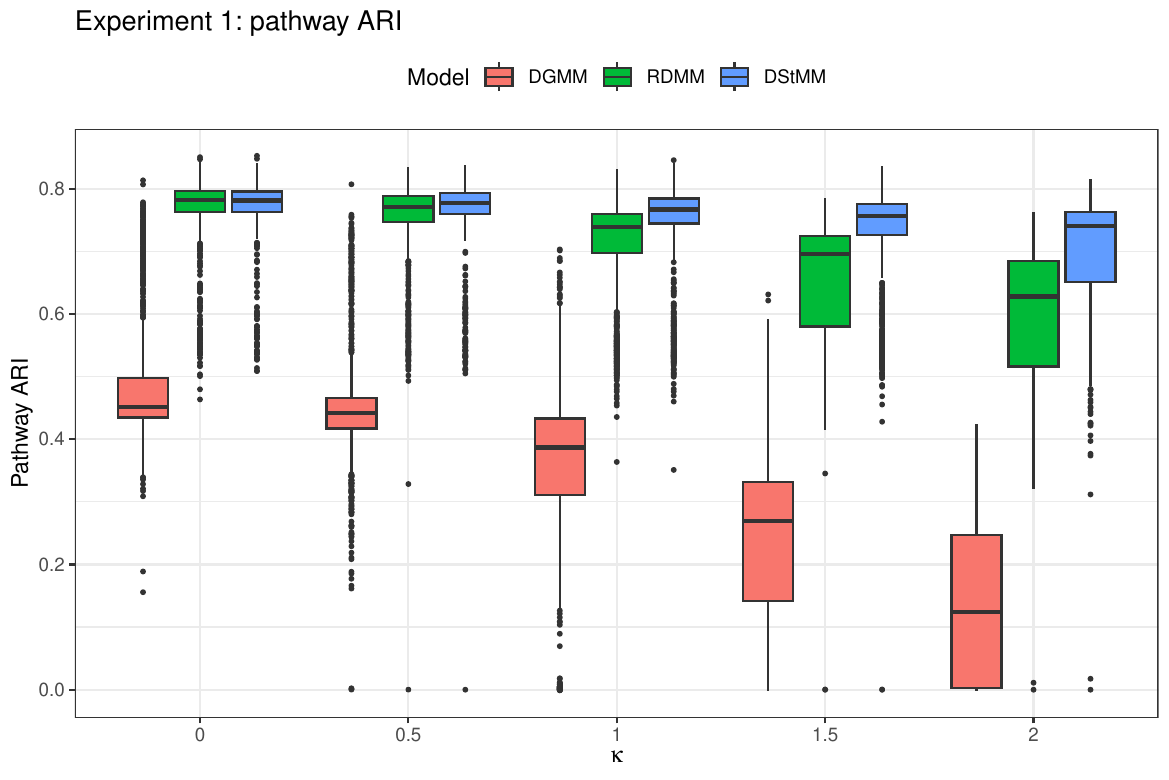}
    \caption{Complete-pathway ARI.}
    \label{fig:exp1-path-ari}
\end{subfigure}

\medskip
\begin{subfigure}[t]{0.65\linewidth}
    \centering
    \includegraphics[width=\linewidth]{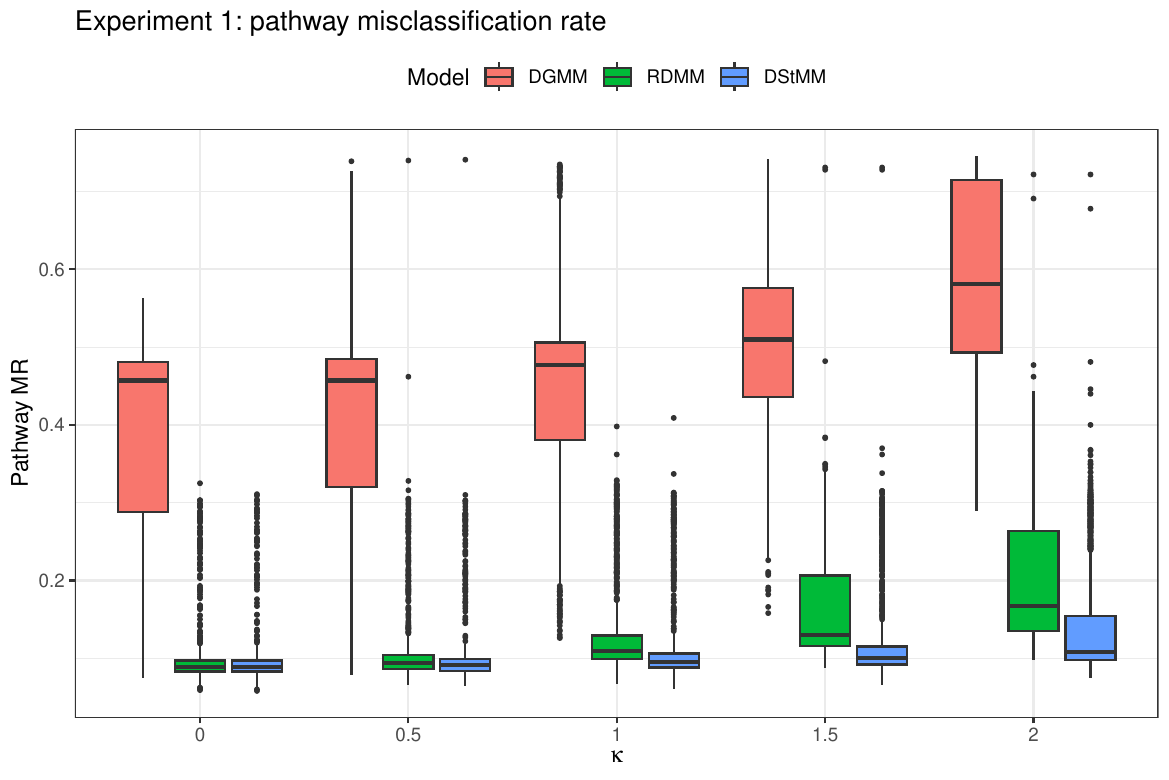}
    \caption{Complete-pathway MR.}
    \label{fig:exp1-path-mr}
\end{subfigure}
\caption{Experiment 1 complete-pathway clustering performance over 1000 Monte Carlo replications as directional skewness increases.}
\label{fig:exp1-pathway-composite}
\end{figure}

Table~\ref{tab:exp1-path-ari-main} and Figures~\ref{fig:exp1-first-layer-composite} and~\ref{fig:exp1-pathway-composite} establish the main nesting result. Under symmetry ($\kappa=0$), RDMM and DStMM recover the complete pathways almost equally well, with mean ARIs of 0.767 and 0.770. Allowing skewness therefore preserves the strong clustering performance of the symmetric submodel. The same conclusion is visible at the first layer, where both heavy-tailed models achieve ARI above 0.95 and MR close to 0.01.

The difference between RDMM and DStMM emerges progressively as the skewness signal strengthens. The DStMM-minus-RDMM pathway-ARI gain is 0.013 at $\kappa=0.5$, 0.034 at $\kappa=1$, 0.066 at $\kappa=1.5$, and 0.094 at $\kappa=2$. The corresponding pathway MR at $\kappa=2$ decreases from 0.201 under RDMM to 0.145 under DStMM. In contrast, their first-layer classifications remain very similar: at $\kappa=2$, the first-layer ARIs are 0.941 for RDMM and 0.947 for DStMM. This separation between first-layer and pathway performance is important. It shows that the additional flexibility of DStMM is expressed most clearly in deeper pathway allocation, while the coarse top-level partition remains highly stable. The skewness term therefore contributes most strongly where the hierarchy distinguishes substructure within the broad first-layer groups.

DGMM behaves differently because it lacks both random scale mixing and directional mean--variance coupling. Its mean pathway ARI falls from 0.489 at $\kappa=0$ to 0.132 at $\kappa=2$, while its first-layer ARI falls from 0.926 to 0.210. Thus, even at $\kappa=0$, where skewness itself is absent, the heavy-tailed setting highlights the value of the shared-scale construction. The comparison among the three models therefore separates two effects: RDMM supplies heavy-tail robustness through the shared scale variable, whereas DStMM adds the further improvement associated with scale-dependent directional displacement.

The parameter estimates support the same interpretation. Across the Experiment 1 grid, DStMM has degrees-of-freedom bias between 0.024 and 0.062 in magnitude and skewness RMSE between 0.096 and 0.106. By comparison, the RDMM bias in $\widehat\nu$ becomes increasingly negative as $\kappa$ grows, from $-0.032$ at symmetry to $-0.190$ at $\kappa=2$. A natural interpretation is that the symmetric RDMM compensates for unmodelled asymmetry by estimating heavier tails, whereas DStMM can allocate the two features to separate parameters. Table~\ref{tab:exp1-details} gives the complete clustering and recovery summaries, and Figure~\ref{fig:exp1-recovery-composite} shows the corresponding recovery distributions.

\begin{table}[!ht]
\centering
\caption{Experiment 1 detailed numerical results: (a) clustering performance and (b) recovery of the common degrees of freedom and first-layer skewness vectors.}
\label{tab:exp1-details}
\begin{subtable}[t]{\linewidth}
\centering
\caption{Clustering performance. Entries are Monte Carlo mean (standard deviation) over successful fits.}
\label{tab:exp1-details-clustering}
\setlength{\tabcolsep}{3pt}
\begin{tabular}{clrrrr}
\toprule
$\kappa$ & Model & L1 ARI & L1 MR & Path ARI & Path MR \\
\midrule
0 & DGMM & 0.926 (0.052) & 0.019 (0.015) & 0.489 (0.102) & 0.378 (0.134) \\
0 & DStMM & 0.953 (0.014) & 0.012 (0.004) & 0.770 (0.054) & 0.097 (0.037) \\
0 & RDMM & 0.958 (0.013) & 0.011 (0.003) & 0.767 (0.061) & 0.100 (0.042) \\
0.5 & DGMM & 0.886 (0.115) & 0.031 (0.040) & 0.456 (0.101) & 0.395 (0.123) \\
0.5 & DStMM & 0.952 (0.033) & 0.012 (0.015) & 0.761 (0.067) & 0.103 (0.048) \\
0.5 & RDMM & 0.955 (0.033) & 0.012 (0.015) & 0.748 (0.074) & 0.110 (0.052) \\
1 & DGMM & 0.720 (0.238) & 0.088 (0.102) & 0.363 (0.123) & 0.445 (0.113) \\
1 & DStMM & 0.951 (0.015) & 0.012 (0.004) & 0.742 (0.079) & 0.115 (0.057) \\
1 & RDMM & 0.951 (0.014) & 0.012 (0.003) & 0.708 (0.082) & 0.133 (0.060) \\
1.5 & DGMM & 0.436 (0.291) & 0.208 (0.154) & 0.233 (0.140) & 0.520 (0.120) \\
1.5 & DStMM & 0.948 (0.045) & 0.014 (0.021) & 0.722 (0.091) & 0.127 (0.068) \\
1.5 & RDMM & 0.945 (0.045) & 0.014 (0.021) & 0.656 (0.093) & 0.164 (0.072) \\
2 & DGMM & 0.210 (0.233) & 0.327 (0.151) & 0.132 (0.126) & 0.594 (0.117) \\
2 & DStMM & 0.947 (0.045) & 0.014 (0.022) & 0.695 (0.102) & 0.145 (0.079) \\
2 & RDMM & 0.941 (0.045) & 0.015 (0.021) & 0.601 (0.097) & 0.201 (0.081) \\
\bottomrule
\end{tabular}
\end{subtable}

\medskip
\begin{subtable}[t]{\linewidth}
\centering
\caption{Parameter recovery. The final column is mean (standard deviation).}
\label{tab:exp1-details-params}
\setlength{\tabcolsep}{3pt}
\begin{tabular}{crrrrr}
\toprule
& \multicolumn{2}{c}{RDMM $\widehat\nu$} & \multicolumn{2}{c}{DStMM $\widehat\nu$} & DStMM \\
\cmidrule(lr){2-3}\cmidrule(lr){4-5}
$\kappa$ & Bias & RMSE & Bias & RMSE & $\mathrm{RMSE}_{\delta}$ \\
\midrule
0 & -0.032 & 0.326 & 0.041 & 0.334 & 0.096 (0.021) \\
0.5 & -0.048 & 0.335 & 0.040 & 0.342 & 0.098 (0.024) \\
1 & -0.067 & 0.333 & 0.062 & 0.346 & 0.099 (0.022) \\
1.5 & -0.152 & 0.353 & 0.024 & 0.327 & 0.101 (0.030) \\
2 & -0.190 & 0.375 & 0.034 & 0.338 & 0.106 (0.024) \\
\bottomrule
\end{tabular}
\end{subtable}

\end{table}

\begin{figure}[!ht]
\centering
\begin{subfigure}[t]{0.48\linewidth}
    \centering
    \includegraphics[width=\linewidth]{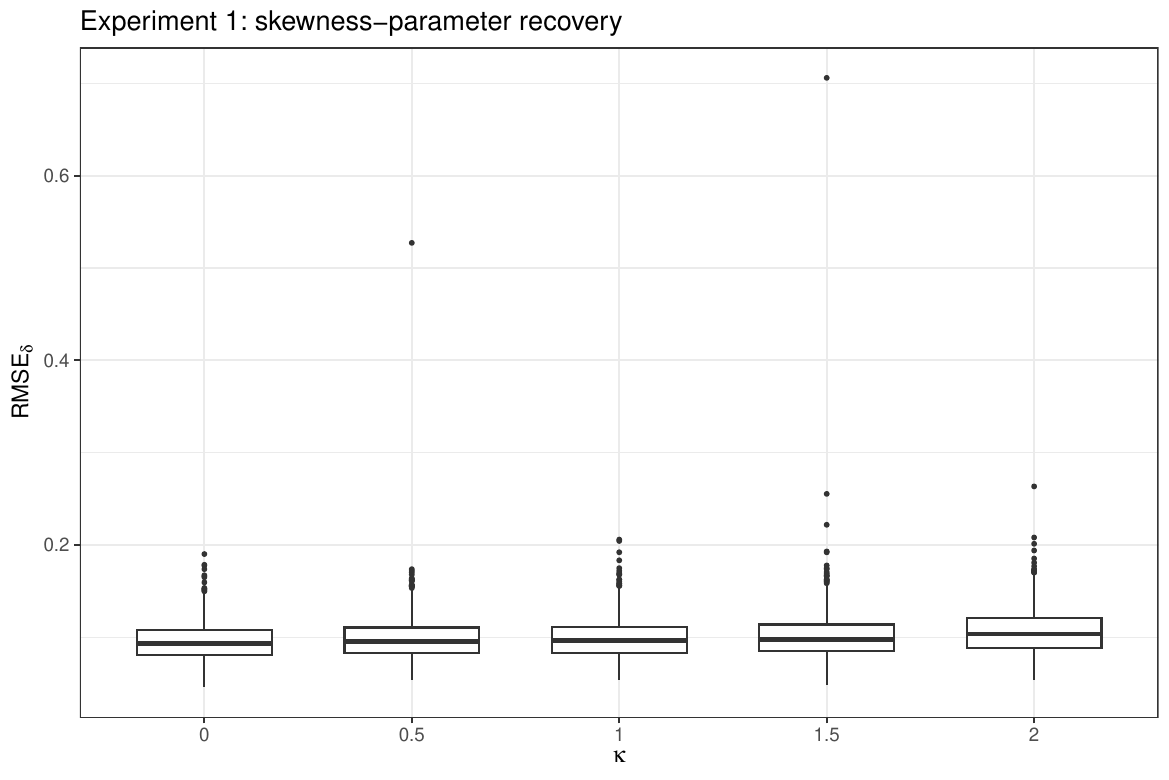}
    \caption{Skewness-vector RMSE.}
    \label{fig:exp1-delta-rmse}
\end{subfigure}
\hfill
\begin{subfigure}[t]{0.48\linewidth}
    \centering
    \includegraphics[width=\linewidth]{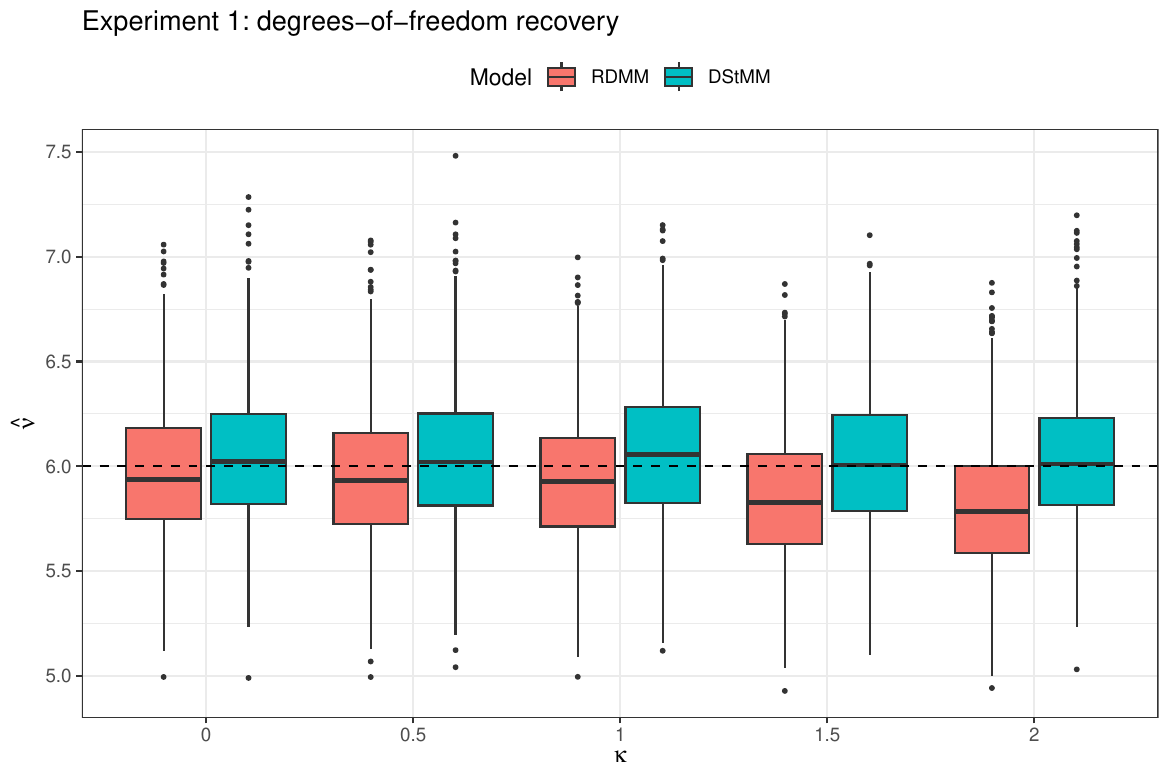}
    \caption{Recovery of $\nu$.}
    \label{fig:exp1-nu-recovery}
\end{subfigure}
\caption{Experiment 1 parameter-recovery diagnostics.}
\label{fig:exp1-recovery-composite}
\end{figure}

The full metrics in Table~\ref{tab:exp1-details}(a) confirm that ARI and MR tell the same substantive story: first-layer error remains very small for both heavy-tailed models, whereas pathway error separates increasingly as $\kappa$ grows. Table~\ref{tab:exp1-details}(b) and Figure~\ref{fig:exp1-recovery-composite} also show that DStMM estimates $\nu$ with little systematic bias over the skewness grid, while RDMM increasingly underestimates $\nu$. The skewness-vector RMSE rises only modestly from 0.096 to 0.106, indicating that the added skewness parameters are estimable at $n=1000$ and prevent the tail parameter from being forced to explain directional asymmetry.

\subsection{Numerical simulation 2}
\label{sec:sim2}

The second experiment uses the factorial grid
\begin{equation*}
\nu\in\{5,8,15\},
\qquad
\kappa\in\{0,1,2\},
\end{equation*}
with $n=1000$ and $R=1000$ in each of the nine cells. The values of $\nu$ move from pronounced heavy tails to more concentrated scale mixing, whereas the values of $\kappa$ move from symmetry to strong directional skewness. Since the same $W$ controls both variance inflation and the skewness displacement, this factorial design tests whether the benefit of modelling asymmetry depends on how variable the common latent scale is.

\begin{table}[t]
\centering
\caption{Experiment 2 mean complete-pathway ARI gain of DStMM over RDMM. Positive values favour DStMM.}
\label{tab:exp2-ari-gain-main}
\begin{tabular}{lccc}
\toprule
& \multicolumn{3}{c}{$\kappa$} \\
\cmidrule(lr){2-4}
$\nu$ & 0 & 1 & 2 \\
\midrule
5 & 0.003 & 0.045 & 0.115 \\
8 & 0.005 & 0.024 & 0.064 \\
15 & 0.003 & 0.010 & 0.028 \\
\bottomrule
\end{tabular}
\end{table}

\begin{figure}[!ht]
\centering
\begin{subfigure}[t]{0.76\linewidth}
    \centering
    \includegraphics[width=\linewidth]{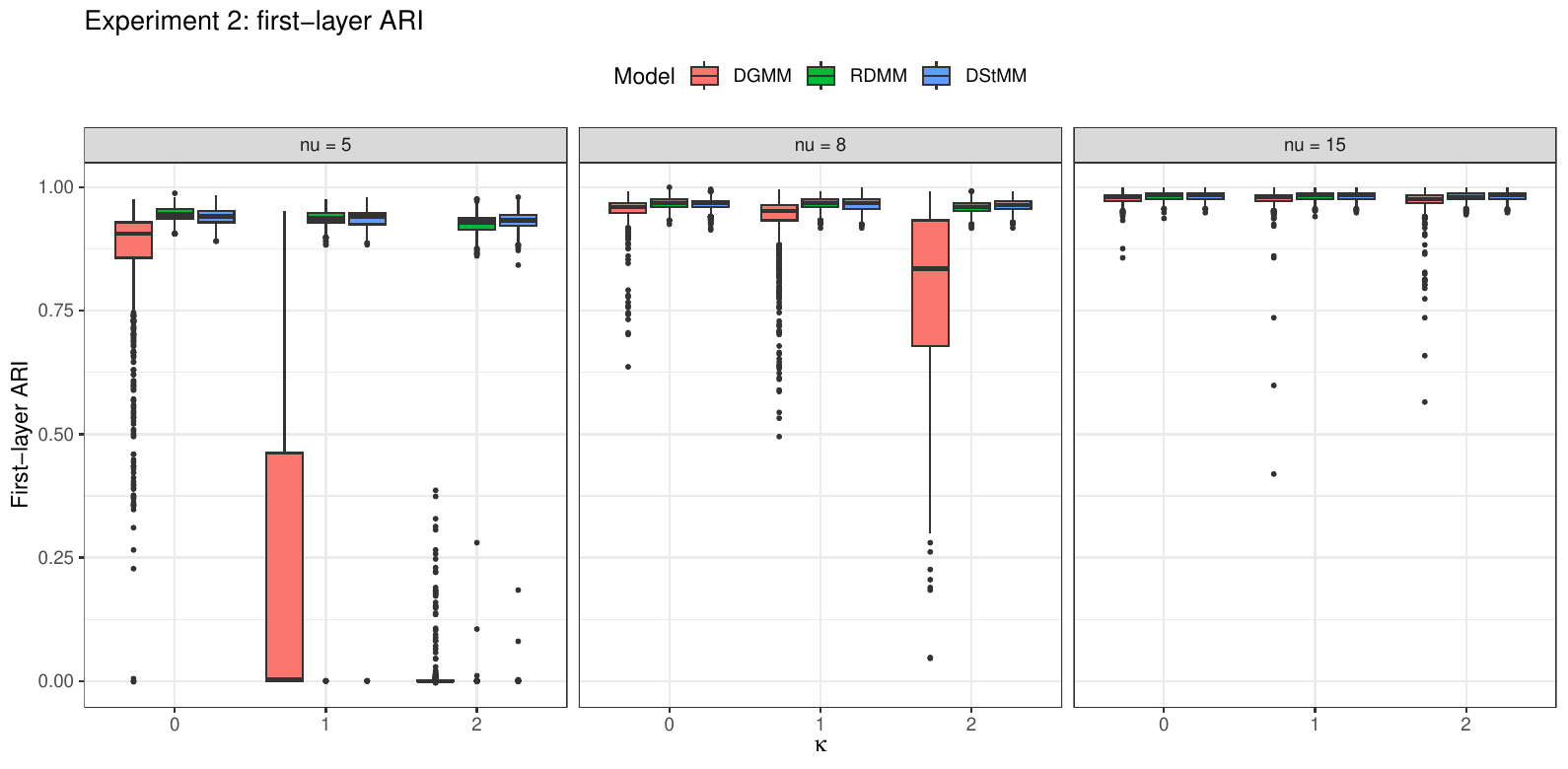}
    \caption{First-layer ARI.}
    \label{fig:exp2-l1-ari}
\end{subfigure}

\medskip
\begin{subfigure}[t]{0.76\linewidth}
    \centering
    \includegraphics[width=\linewidth]{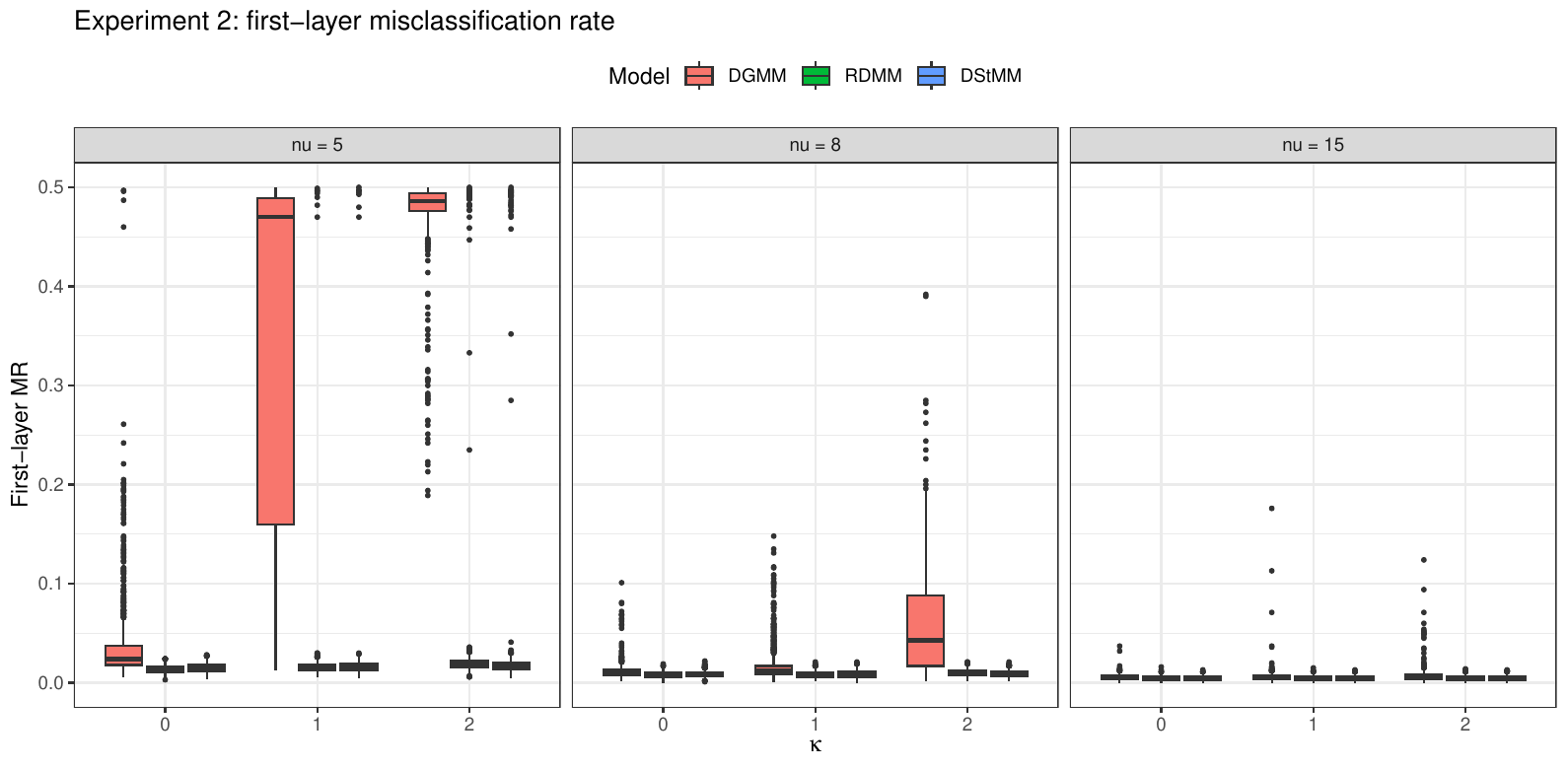}
    \caption{First-layer MR.}
    \label{fig:exp2-l1-mr}
\end{subfigure}
\caption{Experiment 2 first-layer clustering performance across the $(\nu,\kappa)$ factorial design.}
\label{fig:exp2-first-layer-composite}
\end{figure}

\begin{figure}[tbp]
\centering
\begin{subfigure}[t]{0.76\linewidth}
    \centering
    \includegraphics[width=\linewidth]{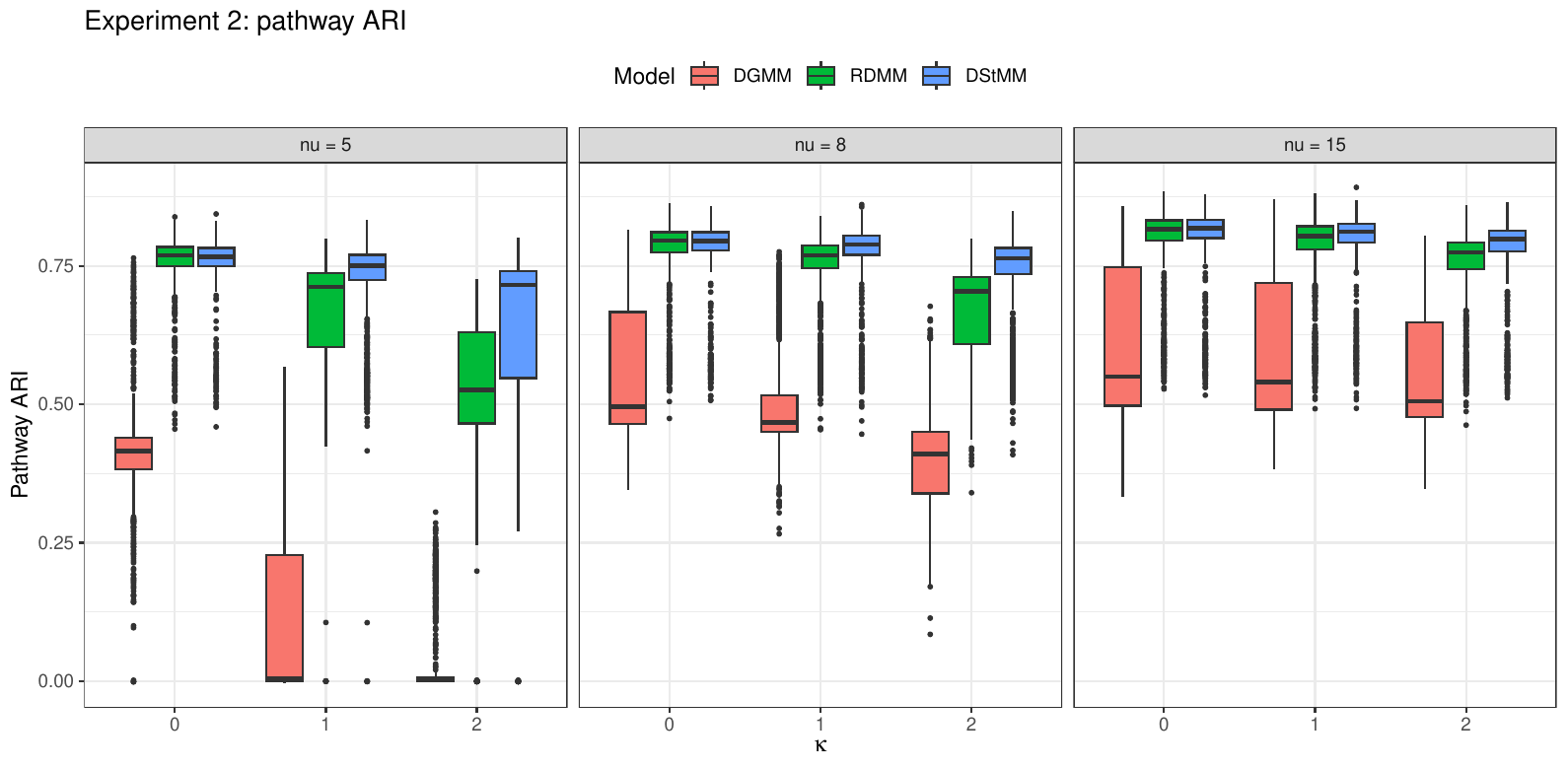}
    \caption{Complete-pathway ARI.}
    \label{fig:exp2-path-ari}
\end{subfigure}

\medskip
\begin{subfigure}[t]{0.76\linewidth}
    \centering
    \includegraphics[width=\linewidth]{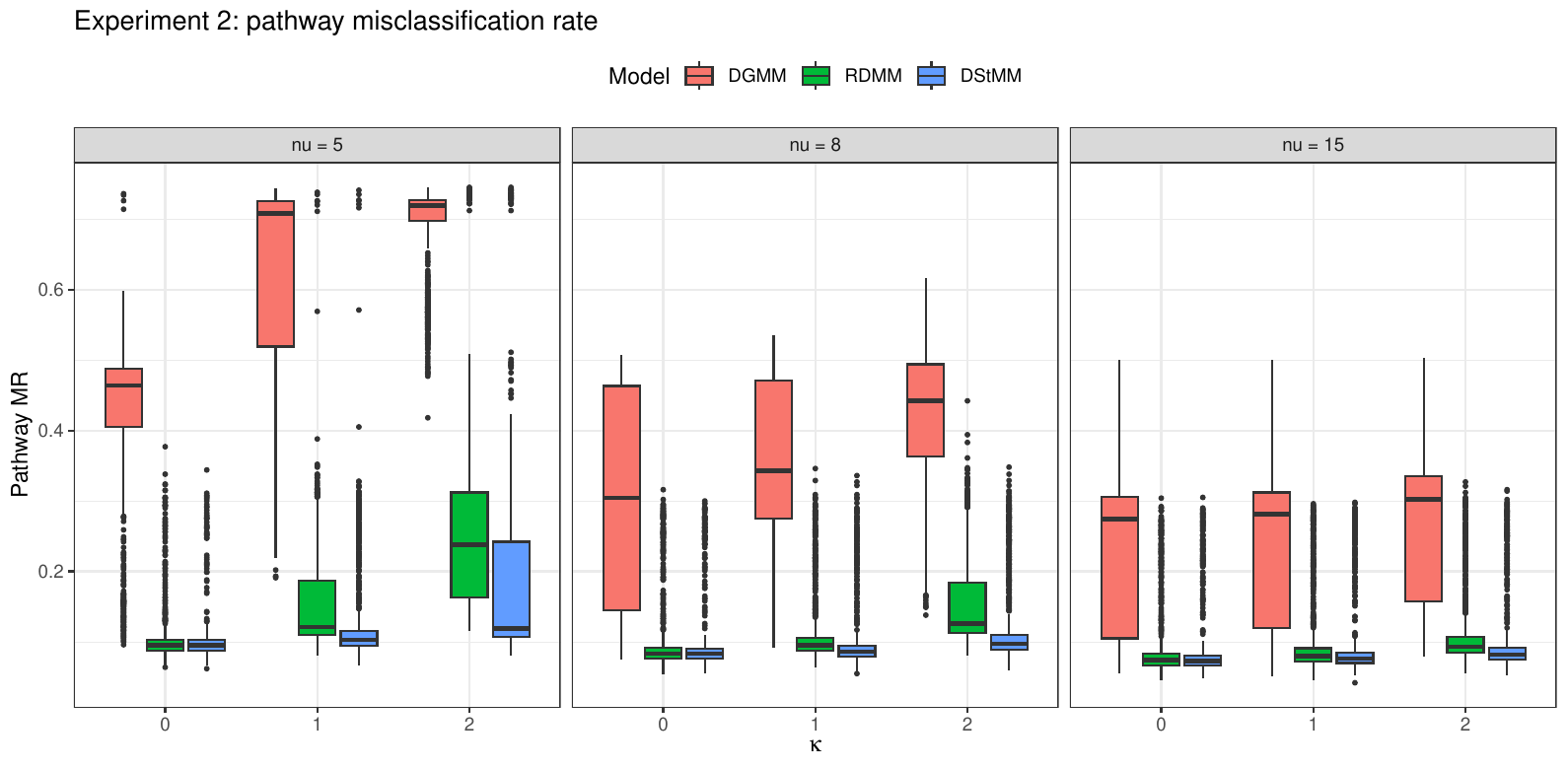}
    \caption{Complete-pathway MR.}
    \label{fig:exp2-path-mr}
\end{subfigure}
\caption{Experiment 2 complete-pathway clustering performance across the $(\nu,\kappa)$ factorial design. The RDMM--DStMM contrast is strongest when large skewness is combined with heavy tails.}
\label{fig:exp2-pathway-composite}
\end{figure}

Table~\ref{tab:exp2-ari-gain-main} and Figures~\ref{fig:exp2-first-layer-composite} and~\ref{fig:exp2-pathway-composite} show a clear interaction between tail weight and asymmetry. In the symmetry column ($\kappa=0$), the mean DStMM-minus-RDMM pathway-ARI difference is only 0.003--0.005 for all three values of $\nu$, again confirming that the extra skewness parameters preserve clustering performance when the data are symmetric. Once $\kappa>0$, however, the gain grows as the tails become heavier. At $\kappa=2$, the pathway-ARI gain is 0.115 for $\nu=5$, 0.064 for $\nu=8$, and 0.028 for $\nu=15$. The same ordering appears at $\kappa=1$. This interaction follows directly from the model construction. The skewness displacement is $W\bm\alpha$, so variation in $W$ controls not only tail inflation but also how strongly observations are displaced in the skewness direction. When $\nu$ is small, the inverse-gamma mixing distribution is more variable, and the explicit DStMM coupling between tail weight and directional displacement becomes especially useful. As $\nu$ increases, $W$ is more concentrated, reducing the effective distinction between a scale-dependent skewness displacement and a more conventional location shift. The empirical reduction in the DStMM--RDMM gap as $\nu$ increases is therefore consistent with the pathway-level mean--variance mixture that motivates the model.

The strongly heavy-tailed and asymmetric cell $(\nu,\kappa)=(5,2)$ makes the practical consequence especially clear: mean pathway ARI is 0.644 for DStMM, 0.529 for RDMM, and 0.042 for DGMM, with pathway MR equal to 0.182, 0.259, and 0.686, respectively. At the lighter-tailed setting $(15,2)$, the corresponding pathway ARIs are 0.775, 0.747, and 0.553. Meanwhile, the first-layer ARIs of RDMM and DStMM remain nearly identical across most of the factorial design. As in Experiment 1, the main value of explicitly modelling skewness is therefore seen in complete-path recovery rather than in a large change to the coarse first-layer partition. The DGMM pattern further distinguishes tail and skewness misspecification. At $\nu=15$, its first-layer ARI remains at least 0.973 even when $\kappa=2$, whereas its pathway ARI is only 0.553. At $\nu=5$, both levels deteriorate sharply with increasing skewness, reaching first-layer ARI 0.007 and pathway ARI 0.042 at $\kappa=2$. Thus, the Gaussian model can retain the broad partition when scale mixing is relatively mild, while the DStMM shows its largest advantage when heavy tails and skewness operate together. Table~\ref{tab:exp2-details} reports the complete Experiment 2 metrics and parameter-recovery summaries. The corresponding recovery distributions are shown in Figure~\ref{fig:exp2-recovery-composite} in Appendix~\ref{app:simulation-diagnostics}.

\begin{table}[!ht]
\centering
\caption{Experiment 2 detailed numerical results across the $(\nu,\kappa)$ factorial design: (a) clustering performance and (b) recovery of the common degrees of freedom and first-layer skewness vectors.}
\label{tab:exp2-details}
\begin{subtable}[t]{\linewidth}
\centering
\caption{Clustering performance. Entries are Monte Carlo mean (standard deviation) over successful fits.}
\label{tab:exp2-details-clustering}
\renewcommand{\arraystretch}{0.90}
\begin{tabular}{cclrrrr}
\toprule
$\nu$ & $\kappa$ & Model & L1 ARI & L1 MR & Path ARI & Path MR \\
\midrule
5 & 0 & DGMM & 0.863 (0.128) & 0.038 (0.045) & 0.423 (0.107) & 0.421 (0.115) \\
5 & 0 & DStMM & 0.941 (0.015) & 0.015 (0.004) & 0.758 (0.050) & 0.102 (0.035) \\
5 & 0 & RDMM & 0.946 (0.014) & 0.014 (0.004) & 0.755 (0.058) & 0.105 (0.041) \\
5 & 1 & DGMM & 0.230 (0.290) & 0.335 (0.175) & 0.119 (0.144) & 0.616 (0.131) \\
5 & 1 & DStMM & 0.931 (0.080) & 0.019 (0.040) & 0.715 (0.101) & 0.131 (0.078) \\
5 & 1 & RDMM & 0.931 (0.080) & 0.019 (0.040) & 0.670 (0.102) & 0.156 (0.079) \\
5 & 2 & DGMM & 0.007 (0.037) & 0.477 (0.039) & 0.042 (0.078) & 0.686 (0.069) \\
5 & 2 & DStMM & 0.902 (0.169) & 0.033 (0.084) & 0.644 (0.153) & 0.182 (0.130) \\
5 & 2 & RDMM & 0.895 (0.167) & 0.034 (0.084) & 0.529 (0.130) & 0.259 (0.123) \\
8 & 0 & DGMM & 0.954 (0.031) & 0.012 (0.008) & 0.553 (0.117) & 0.298 (0.140) \\
8 & 0 & DStMM & 0.966 (0.011) & 0.009 (0.003) & 0.777 (0.067) & 0.097 (0.048) \\
8 & 0 & RDMM & 0.969 (0.011) & 0.008 (0.003) & 0.773 (0.072) & 0.100 (0.051) \\
8 & 1 & DGMM & 0.932 (0.064) & 0.018 (0.018) & 0.501 (0.095) & 0.343 (0.126) \\
8 & 1 & DStMM & 0.965 (0.013) & 0.009 (0.003) & 0.770 (0.067) & 0.099 (0.046) \\
8 & 1 & RDMM & 0.966 (0.012) & 0.009 (0.003) & 0.747 (0.071) & 0.111 (0.049) \\
8 & 2 & DGMM & 0.790 (0.165) & 0.059 (0.051) & 0.401 (0.080) & 0.419 (0.092) \\
8 & 2 & DStMM & 0.964 (0.012) & 0.009 (0.003) & 0.732 (0.083) & 0.122 (0.060) \\
8 & 2 & RDMM & 0.961 (0.012) & 0.010 (0.003) & 0.668 (0.084) & 0.155 (0.062) \\
15 & 0 & DGMM & 0.979 (0.011) & 0.005 (0.003) & 0.608 (0.129) & 0.231 (0.116) \\
15 & 0 & DStMM & 0.982 (0.009) & 0.005 (0.002) & 0.793 (0.076) & 0.091 (0.053) \\
15 & 0 & RDMM & 0.983 (0.009) & 0.004 (0.002) & 0.790 (0.077) & 0.093 (0.053) \\
15 & 1 & DGMM & 0.977 (0.026) & 0.006 (0.007) & 0.592 (0.122) & 0.241 (0.114) \\
15 & 1 & DStMM & 0.982 (0.009) & 0.004 (0.002) & 0.787 (0.076) & 0.095 (0.056) \\
15 & 1 & RDMM & 0.983 (0.009) & 0.004 (0.002) & 0.777 (0.079) & 0.100 (0.056) \\
15 & 2 & DGMM & 0.973 (0.028) & 0.007 (0.008) & 0.553 (0.106) & 0.275 (0.116) \\
15 & 2 & DStMM & 0.982 (0.009) & 0.005 (0.002) & 0.775 (0.074) & 0.099 (0.053) \\
15 & 2 & RDMM & 0.981 (0.009) & 0.005 (0.002) & 0.747 (0.074) & 0.112 (0.052) \\
\bottomrule
\end{tabular}
\end{subtable}

\medskip
\begin{subtable}[t]{\linewidth}
\centering
\caption{Parameter recovery. The final column is mean (standard deviation).}
\label{tab:exp2-details-params}
\renewcommand{\arraystretch}{0.90}
\begin{tabular}{ccrrrrr}
\toprule
& & \multicolumn{2}{c}{RDMM $\widehat\nu$} & \multicolumn{2}{c}{DStMM $\widehat\nu$} & DStMM \\
\cmidrule(lr){3-4}\cmidrule(lr){5-6}
$\nu$ & $\kappa$ & Bias & RMSE & Bias & RMSE & $\mathrm{RMSE}_{\delta}$ \\
\midrule
5 & 0 & -0.039 & 0.269 & 0.014 & 0.272 & 0.082 (0.018) \\
5 & 1 & -0.094 & 0.272 & 0.010 & 0.257 & 0.087 (0.043) \\
5 & 2 & -0.135 & 0.290 & 0.042 & 0.262 & 0.105 (0.082) \\
8 & 0 & -0.040 & 0.478 & 0.092 & 0.499 & 0.121 (0.027) \\
8 & 1 & -0.098 & 0.479 & 0.093 & 0.497 & 0.123 (0.027) \\
8 & 2 & -0.193 & 0.516 & 0.116 & 0.526 & 0.131 (0.031) \\
15 & 0 & 0.007 & 1.188 & 0.627 & 1.458 & 0.199 (0.044) \\
15 & 1 & 0.039 & 1.245 & 0.707 & 1.581 & 0.201 (0.043) \\
15 & 2 & -0.232 & 1.183 & 0.534 & 1.426 & 0.205 (0.043) \\
\bottomrule
\end{tabular}
\end{subtable}
\end{table}

Table~\ref{tab:exp2-details}(a) confirms that the tail--skewness interaction is reproduced by both ARI and MR. For example, at $(\nu,\kappa)=(5,2)$ DStMM reduces pathway MR by 0.077 relative to RDMM, whereas at $(15,2)$ the reduction is 0.013; at $\kappa=0$, the two heavy-tailed models remain nearly indistinguishable for every $\nu$. The parameter-recovery summaries also follow the geometry of the model as $\nu$ varies. For DStMM, $\operatorname{RMSE}_{\delta}$ is roughly 0.08--0.11 at $\nu=5$ and about 0.20 at $\nu=15$, while the RMSE of $\widehat\nu$ changes from about 0.26--0.27 to 1.43--1.58 as the latent scale $W$ becomes increasingly concentrated near one. Across these settings, the clustering results remain strong and preserve the same tail--skewness interaction.

\subsection{Robustness checks}
\label{sec:sim-robustness}

The principal conclusions are not confined to the baseline $n=1000$ balanced design. Table~\ref{tab:exp1-robustness} summarises the smaller-sample and unbalanced-mixture checks, while Figure~\ref{fig:exp1-robustness-pathway} shows the complete-pathway distributions. The corresponding first-layer distributions are provided in Figure~\ref{fig:exp1-robustness-first-layer} in Appendix~\ref{app:simulation-diagnostics}.

\begin{table}[!ht]
\centering
\caption{Experiment 1 robustness checks: complete-pathway ARI, reported as Monte Carlo mean (standard deviation). Panel (a) reduces the sample size to $n=500$; panel (b) uses $\bm\pi^{(1)}=\bm\pi^{(2)}=(0.7,0.3)^\top$.}
\label{tab:exp1-robustness}
\begin{subtable}[t]{0.49\linewidth}
\centering
\caption{$n=500$.}
\label{tab:exp1-rob-n500}
\begin{tabular}{crrr}
\toprule
$\kappa$ & DGMM & RDMM & DStMM \\
\midrule
0 & 0.464 (0.104) & 0.731 (0.081) & 0.734 (0.074) \\
0.5 & 0.435 (0.095) & 0.711 (0.088) & 0.724 (0.080) \\
1 & 0.378 (0.102) & 0.677 (0.091) & 0.710 (0.089) \\
1.5 & 0.292 (0.126) & 0.604 (0.111) & 0.664 (0.118) \\
2 & 0.212 (0.132) & 0.543 (0.110) & 0.625 (0.127) \\
\bottomrule
\end{tabular}
\end{subtable}\hfill
\begin{subtable}[t]{0.49\linewidth}
\centering
\caption{Unbalanced layer proportions.}
\label{tab:exp1-rob-unbal}
\setlength{\tabcolsep}{1.8pt}
\begin{tabular}{crrr}
\toprule
$\kappa$ & DGMM & RDMM & DStMM \\
\midrule
0 & 0.373 (0.120) & 0.759 (0.047) & 0.759 (0.044) \\
0.5 & 0.286 (0.114) & 0.762 (0.047) & 0.768 (0.043) \\
1 & 0.240 (0.092) & 0.751 (0.057) & 0.776 (0.044) \\
1.5 & 0.216 (0.076) & 0.708 (0.088) & 0.768 (0.066) \\
2 & 0.193 (0.065) & 0.640 (0.120) & 0.758 (0.079) \\
\bottomrule
\end{tabular}
\end{subtable}

\end{table}

\begin{figure}[!ht]
\centering
\begin{subfigure}[t]{0.76\linewidth}
    \centering
    \includegraphics[width=\linewidth]{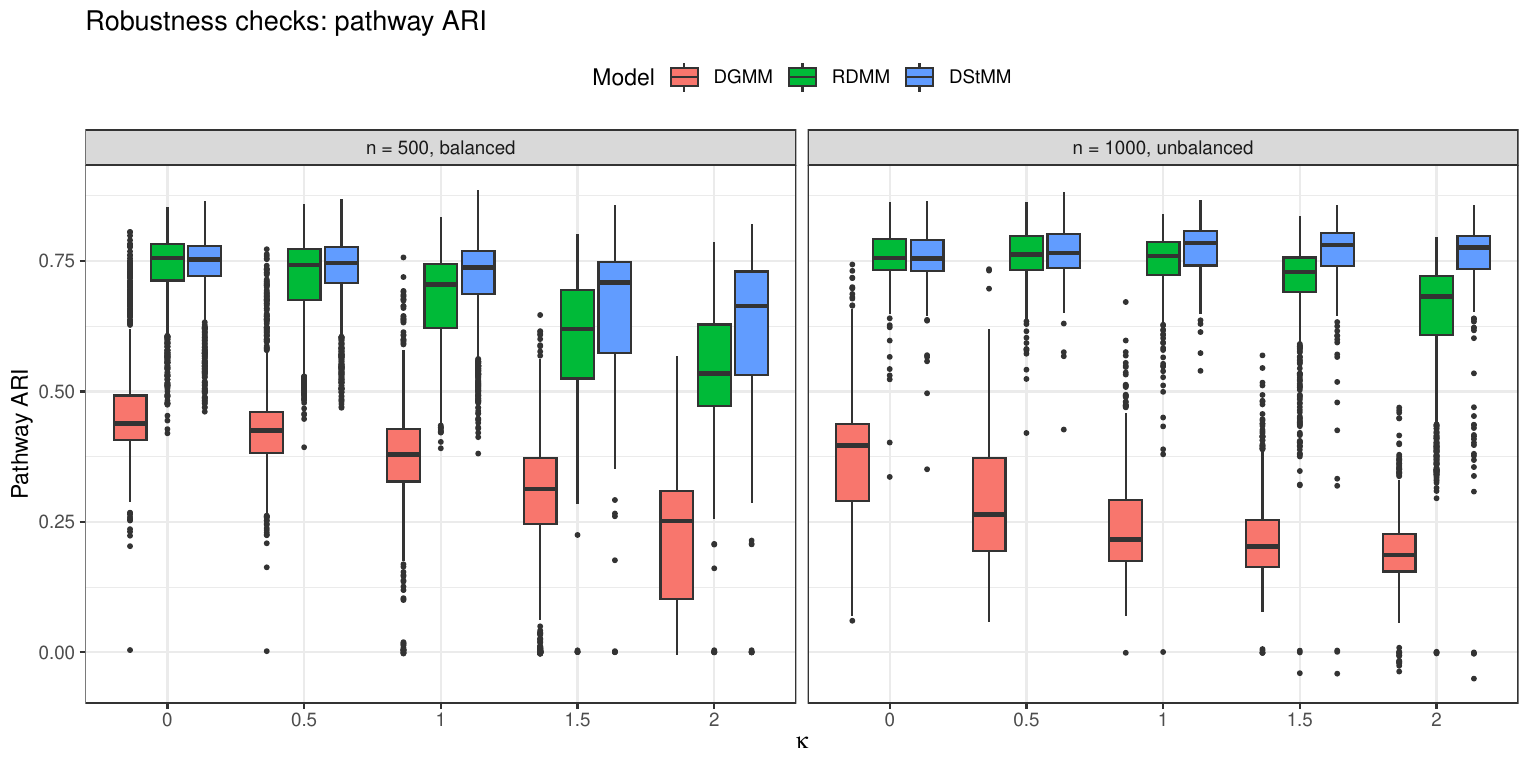}
    \caption{Complete-pathway ARI.}
    \label{fig:rob-path-ari}
\end{subfigure}

\medskip
\begin{subfigure}[t]{0.76\linewidth}
    \centering
    \includegraphics[width=\linewidth]{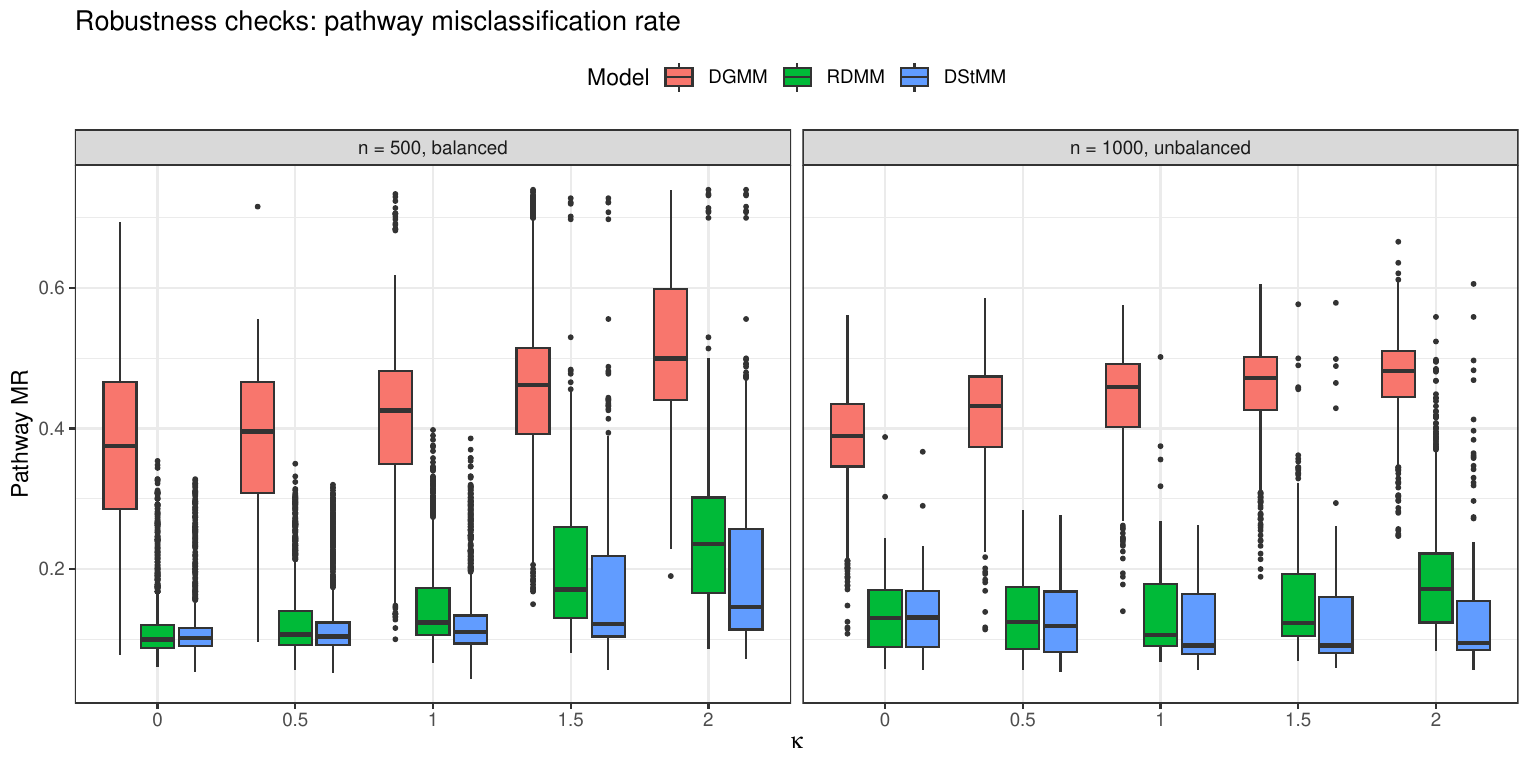}
    \caption{Complete-pathway MR.}
    \label{fig:rob-path-mr}
\end{subfigure}
\caption{Experiment 1 robustness checks at the complete-pathway level for the smaller-sample ($n=500$) and unbalanced-mixture designs.}
\label{fig:exp1-robustness-pathway}
\end{figure}

Table~\ref{tab:exp1-robustness}(a) and Figure~\ref{fig:exp1-robustness-pathway} show that the model ordering is preserved at $n=500$. At $\kappa=2$, mean pathway ARI is 0.625 for DStMM and 0.543 for RDMM, while pathway MR is 0.191 and 0.245, respectively. The DStMM--RDMM ARI advantage of 0.082 is close to the 0.094 gain in the main $n=1000$ design, showing that the benefit of modelling skewness remains clear at the smaller sample size. Figure~\ref{fig:exp1-robustness-first-layer} in Appendix~\ref{app:simulation-diagnostics} further shows stable first-layer performance for RDMM and DStMM. Table~\ref{tab:exp1-robustness}(b) shows the same qualitative conclusion under unequal mixture proportions. At symmetry, RDMM and DStMM have identical mean pathway ARI to three decimals (0.759). At $\kappa=2$, the values separate to 0.640 and 0.758, and pathway MR decreases from 0.189 under RDMM to 0.122 under DStMM. The DStMM advantage therefore persists when one local component is less prevalent, reinforcing the robustness of the skewness-dependent pathway gain. The computational diagnostics show that the added skewness modelling retains very high fitting stability; detailed timing distributions and fitting-success counts are reported in Appendix~\ref{app:simulation-diagnostics}.

\subsection{Overall simulation findings}

Taken together, the two experiments give a coherent practical interpretation. The DStMM provides a unified deep-mixture extension across the regimes considered here. Under symmetry, it retains clustering performance essentially indistinguishable from the RDMM, while increasing directional asymmetry produces progressively larger gains in complete-pathway recovery. These gains are strongest under heavier tails, where the shared scale and directional displacement act most strongly together. Relative to the DGMM, the common scale supplies heavy-tail robustness and the skewness term adds direction-sensitive flexibility. The robustness checks confirm the same qualitative pattern under reduced sample size and unequal mixture proportions. These results also clarify what the additional skewness parameters are doing. Their main impact is not to redefine the top-level groups, for which RDMM and DStMM usually have nearly identical ARI, but to improve discrimination among complete pathways. At the same time, explicitly modelling skewness prevents the degrees-of-freedom parameter from absorbing asymmetry, as seen in the increasingly negative RDMM $\widehat\nu$ bias in Experiment 1. The simulations therefore support the pathway construction on both predictive and inferential grounds: DStMM improves hierarchical classification in the regimes for which it was designed, while retaining the symmetric RDMM as an empirically effective limiting case.

\section{Real applications}
\label{sec:realdata}

\subsection{Case 1: Handwritten Digits}
\label{sec:mfeat}

We consider the \texttt{mfeat-fac} representation of the UCI Multiple Features data set \citep{uciMultipleFeatures}. It contains $n=2000$ handwritten-digit observations from ten classes, with 200 observations per digit, represented by $p=216$ profile-correlation features. We retain all ten classes and all non-degenerate variables and standardise the features before fitting. The known number of digit classes is used only to set the first-layer order $K_1=10$; individual class labels do not enter model fitting and are used only afterward to compute external clustering measures. The broader Multiple Features collection remains a common benchmark in recent clustering work, including deep and interpretable multi-view methods \citep{liu2024dmvgc, jiang2025interpretable}; the present study focuses on the single \texttt{mfeat-fac} representation.

For a controlled comparison of the component distributional assumptions, DGMM, RDMM and DStMM are fitted under the same working architecture $(K_1,K_2,r_1,r_2)=(10,1,10,3).$ Using a common latent architecture keeps the comparison focused on the distributional specification of the three deep models rather than on differences arising from architecture selection. No architecture selection is therefore performed in this application. In the DStMM fit, the first-layer skewness vectors $\bm\delta^{(1)}_a$, $a=1,\ldots,10$, are estimated separately for the ten local components, with no equality constraint across clusters. Numerical fitting follows the workflow in Section~\ref{sec:fitting-workflow}, using the observed-data likelihood, posterior-path allocation, and weighted-regression updates described above. All fitting decisions are made without individual class-label information; the labels are reserved exclusively for the post-fit ARI and MR calculations reported below.

\begin{table}[!ht]
\centering
\caption{Clustering comparisons on the UCI \texttt{mfeat-fac} data. Panel (a) reports the controlled comparison under the common working architecture $(K_1,K_2,r_1,r_2)=(10,1,10,3)$. Panels (b) and (c) report published reference results on the same single-view benchmark for broader empirical context.}
\label{tab:mfeat-comparison}
\renewcommand{\arraystretch}{0.98}
\setlength{\tabcolsep}{6pt}

\textbf{(a) Comparison under the common deep architecture.}
\par\smallskip
\begin{tabular}{llrr}
\toprule
Method & Family & ARI & MR\\
\midrule
DStMM (proposed) & Deep skew-$t$ mixture & \textbf{0.7325} & \textbf{0.2195} \\
RDMM             & Deep $t$ mixture      & 0.7232 & 0.2240 \\
DGMM             & Deep Gaussian mixture & 0.6840 & 0.2355 \\
K-means          & Partitioning          & 0.6186 & 0.2725 \\
Ward.D2          & Hierarchical          & 0.6169 & 0.2845 \\
\bottomrule
\end{tabular}

\medskip
\textbf{(b) Published ARI reference results on the same single-view
\texttt{mfeat-fac} clustering problem.}
\par\smallskip
\begin{tabular}{lrl}
\toprule
Method & ARI & Source\\
\midrule
DStMM (proposed)        & \textbf{0.7325} & This paper\\
SNN-DPC                 & 0.7186 & \citet{long2022clustering}\\
DPC                     & 0.6415 & \citet{long2022clustering}\\
LDP-MST                 & 0.5597 & \citet{long2022clustering}\\
DPC-DBFN                & 0.4885 & \citet{long2022clustering}\\
Ranked $k$-medoids      & 0.4390 & \citet{zadegan2013ranked}\\
Simple/Fast $k$-medoids & 0.4290 & \citet{zadegan2013ranked}\\
KHM                     & 0.2870 & \citet{zadegan2013ranked}\\
\bottomrule
\end{tabular}

\medskip
\textbf{(c) Published MR reference results on the same single-view
\texttt{mfeat-fac} clustering problem.}
\par\smallskip
\begin{tabular}{lrl}
\toprule
Method & MR & Source\\
\midrule
DStMM (proposed) & \textbf{0.2195} & This paper\\
HPFCM            & 0.2253 & \citet{shao2019hybrid}\\
ClusiVAT         & 0.2908 & \citet{shao2019hybrid}\\
NRCG-ONMF        & 0.3528 & \citet{zhang2016efficient}\\
K-means          & 0.3787 & \citet{zhang2016efficient}\\
CAPNMF           & 0.3995 & \citet{zhang2016efficient}\\
OFCM             & 0.3995 & \citet{shao2019hybrid}\\
ONPMF            & 0.4551 & \citet{zhang2016efficient}\\
EMONMF           & 0.4614 & \citet{zhang2016efficient}\\
CVXNMF           & 0.4692 & \citet{zhang2016efficient}\\
NMF-GCD          & 0.5234 & \citet{zhang2016efficient}\\
SPANONMF         & 0.7662 & \citet{zhang2016efficient}\\
\bottomrule
\end{tabular}

\par\smallskip
\footnotesize
For Panel (c), MR is reported as one minus the published correct-classification proportion so that smaller values consistently indicate better clustering. The ARI values in Panel (b) are reproduced from \citet{long2022clustering} and \citet{zadegan2013ranked}, while the MR values in Panel (c) are reproduced from \citet{shao2019hybrid} and \citet{zhang2016efficient}.
\end{table}

Under the common working architecture $(K_1,K_2,r_1,r_2)=(10,1,10,3)$, DStMM gives the strongest clustering performance among the three deep mixture models, with ARI $0.7325$ and MR $0.2195$. RDMM gives ARI $0.7232$ and MR $0.2240$, while DGMM gives ARI $0.6840$ and MR $0.2355$. Thus, under the controlled architecture, the empirical ordering is $\mathrm{DStMM}>\mathrm{RDMM}>\mathrm{DGMM}.$ This ordering is consistent with the simulation findings. Replacing the Gaussian pathway distribution with a heavy-tailed specification improves clustering performance, while allowing component-specific directional asymmetry yields a further improvement over the symmetric RDMM.

The magnitude of the DStMM--RDMM difference is more modest here than in the strongly skewed and heavy-tailed simulation settings, but the direction of the improvement is consistent for both external criteria: relative to RDMM, DStMM increases ARI from $0.7232$ to $0.7325$ and reduces MR from $0.2240$ to $0.2195$. Relative to DGMM, the improvement is more pronounced, with an ARI increase of $0.0485$ and an MR reduction of $0.0160$. These results suggest that both heavy-tail robustness and directional asymmetry contribute to the observed clustering performance on the \texttt{mfeat-fac} representation.

Panels (b) and (c) of Table~\ref{tab:mfeat-comparison} provide broader literature context. The DStMM ARI of $0.7325$ is higher than all published ARI reference values listed in Panel (b), including the $0.7186$ reported for SNN-DPC by \citet{long2022clustering}. Likewise, the DStMM MR of $0.2195$ is lower than all published MR reference values listed in Panel (c), the closest being $0.2253$ for HPFCM reported by \citet{shao2019hybrid}. These comparisons should be interpreted as benchmark context rather than as strictly controlled pairwise comparisons, since the published methods use different modelling, optimisation and preprocessing procedures.

Figure~\ref{fig:mfeat-delta} displays the estimated component-specific first-layer skewness vectors $\widehat{\bm\delta}^{(1)}_a$, $a=1,\ldots,10$, under the common DStMM architecture. Each row corresponds to a fitted first-layer component and each column to a retained \texttt{mfeat-fac} feature. The estimated coefficients exhibit visibly different signed patterns across components, with some components displaying substantially stronger feature-specific departures from symmetry than others. This heterogeneity provides empirical support for estimating component-specific skewness rather than imposing a common skewness direction across the ten first-layer components.

\begin{figure}[!ht]
\centering
\includegraphics[width=0.95\textwidth]{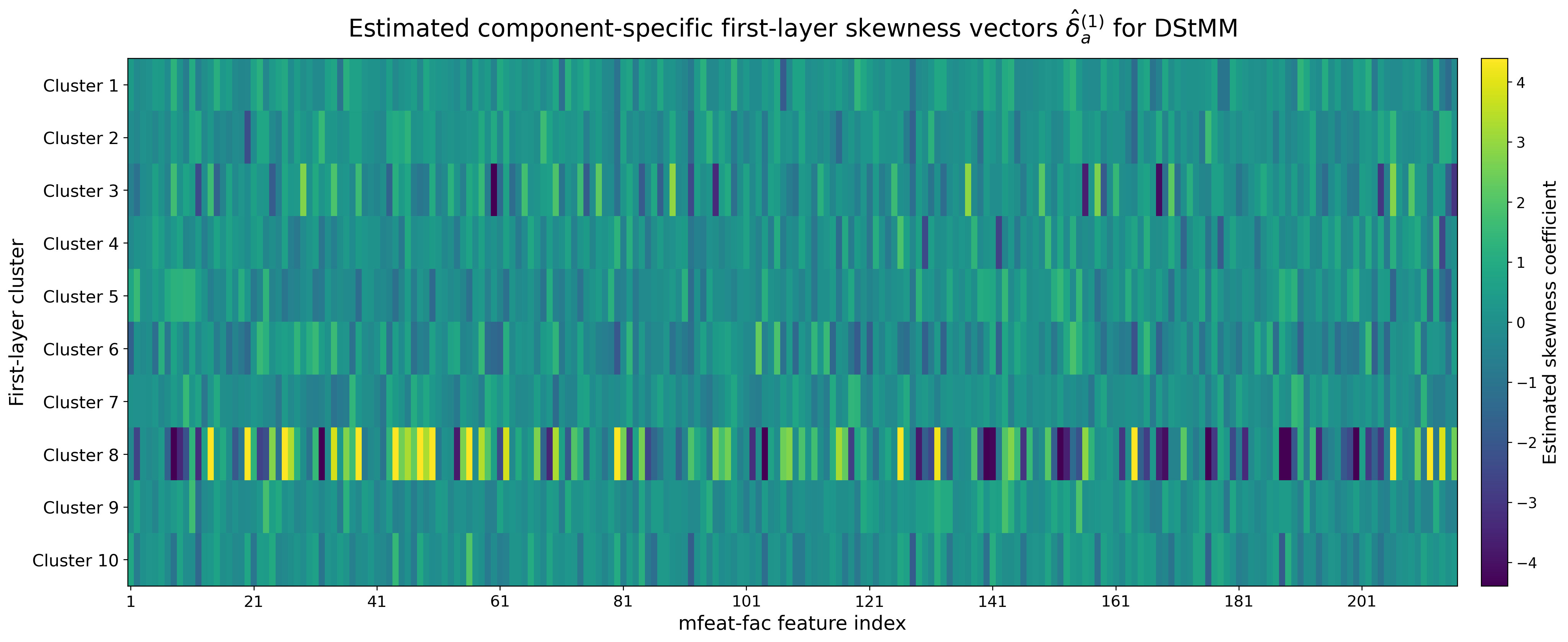}
\caption{Estimated component-specific first-layer skewness vectors $\widehat{\bm\delta}^{(1)}_a$ for the DStMM under $(K_1,K_2,r_1,r_2)=(10,1,10,3)$. Rows correspond to the ten first-layer components and columns to the retained \texttt{mfeat-fac} features. Colour represents the signed estimated skewness coefficient.}
\label{fig:mfeat-delta}
\end{figure}

\subsection{Case 2: Gas Sensor Array Drift Data}

We next consider the UCI Gas Sensor Array Drift at Different Concentrations data set~\citep{rodriguez2014calibration}. It contains $n=13{,}910$ observations represented by $p=128$ sensor-response variables, obtained from 16 chemical sensors exposed to six gases over ten temporal batches. All non-degenerate sensor variables are retained and standardised before fitting. The known number of gas classes is used only to set the first-layer order $K_1=6$; individual gas labels do not enter model fitting or architecture selection and are used only afterward to compute external clustering measures.

The sensor measurements exhibit substantial departures from Gaussianity. Across the 128 variables, the median absolute skewness is $2.377$, and the median excess kurtosis is $10.317$. Moreover, $96.1\%$ of the variables have absolute skewness greater than one and $82.0\%$ have excess kurtosis greater than three. The data therefore provide a natural setting in which to assess the contribution of both heavy-tail modelling and component-specific asymmetry.

In contrast to Case~1, where a common architecture is imposed to isolate the effect of the component distribution, the deeper architecture is selected separately within each model family in this application. We fix $K_1=6$ and consider $K_2\in\{1,2,3,4\}, (r_1,r_2)\in\{(8,2),(12,3),(16,4)\}.$ For each of DGMM, RDMM, and DStMM, the architecture with the smallest BIC returned by the corresponding fitting implementation is retained for the final comparison; for DStMM, the implemented count is given in \eqref{eq:bic-df-new}. The same counting convention is applied across all candidate architectures within each family. Allowing $K_2>1$ also permits the implemented selection criterion to choose a non-trivial second mixture layer when it is favoured within the candidate grid.

\begin{table}[!htbp]
\centering
\caption{Clustering comparisons on the UCI Gas Sensor Array Drift data. Panel (a) reports the architecture selected by the implementation-based BIC criterion within each deep-model family. Panel (b) gives selected published ARI reference results on the same gas-sensor data set for broader empirical comparison.}
\label{tab:gas-comparison}
\renewcommand{\arraystretch}{0.98}
\setlength{\tabcolsep}{6pt}

\textbf{(a) Comparison under the implementation-based BIC selections.}
\par\smallskip
\begin{tabular}{lrrrrrr}
\toprule
Model & $K_1$ & $K_2$ & $r_1$ & $r_2$ & ARI & MR\\
\midrule
DStMM (proposed) & 6 & 3 & 16 & 4 & \textbf{0.2989} & \textbf{0.4717}\\
RDMM             & 6 & 4 & 16 & 4 & 0.2014 & 0.5357\\
DGMM             & 6 & 1 & 12 & 3 & 0.1803 & 0.5438\\
\bottomrule
\end{tabular}

\medskip
\textbf{(b) Published ARI reference results on the same gas-sensor data set.}
\par\smallskip
\begin{tabular}{lrl}
\toprule
Method & ARI & Source\\
\midrule
DStMM (proposed)                    & \textbf{0.2989} & This paper\\
AC Ward (Extended IF + t-SNE)      & 0.2972 & Guyeux et al. (2019)\\
Birch (Extended IF + t-SNE)        & 0.2926 & Guyeux et al. (2019)\\
$K$-means++ (IF + t-SNE)           & 0.2908 & Guyeux et al. (2019)\\
$K$-means (IF + t-SNE)             & 0.2904 & Guyeux et al. (2019)\\
Birch (IF + t-SNE)                 & 0.2820 & Guyeux et al. (2019)\\
GMM (IF + t-SNE)                   & 0.2666 & Guyeux et al. (2019)\\
Ward                                & 0.2007 & Dey et al. (2019)\\
RPHash-TWRP                         & 0.1759 & Dey et al. (2019)\\
$K$-means                           & 0.1539 & Dey et al. (2019)\\
\bottomrule
\end{tabular}

\par\smallskip
\footnotesize Published values in Panel (b) are selected reference results reproduced from the cited studies. IF denotes Isolation Forest. The Guyeux et al. configurations combine outlier preprocessing and t-SNE dimensionality reduction with the indicated clustering procedure, whereas DStMM is fitted directly to the standardised 128-dimensional sensor-response variables.
\end{table}

Among the three deep mixture formulations, Table~\ref{tab:gas-comparison}(a) shows that DStMM gives the strongest clustering performance. Under the corresponding implementation-based BIC criterion, DStMM selects $(K_1,K_2,r_1,r_2)=(6,3,16,4),$ with ARI $0.2989$ and MR $0.4717$. The selected RDMM gives ARI $0.2014$ and MR $0.5357$, while the selected DGMM gives ARI $0.1803$ and MR $0.5438$. Hence the empirical ordering is $\mathrm{DStMM}>\mathrm{RDMM}>\mathrm{DGMM}.$ As in Case~1 and the simulation study, the results indicate a gain from heavy-tail modelling over the Gaussian formulation, followed by a further improvement when component-specific asymmetry is incorporated. The advantage of DStMM is not restricted to the final selected architectures. Under the common architecture $(K_1,K_2,r_1,r_2)=(6,3,16,4),$ for example, DStMM attains ARI $0.2989$ and MR $0.4717$, compared with ARI $0.2037$ and MR $0.5313$ for RDMM. Improvements are also observed for several other combinations of $K_2$ and the latent dimensions, indicating that the gain from the skew-$t$ specification is not driven by a single isolated architecture.

The published reference results in Table~\ref{tab:gas-comparison}(b) provide additional empirical context. \citet{dey2019clustering} considered the same $13{,}910$ observations and 128 sensor variables and reported ARI values of $0.1539$ for $K$-means, $0.1759$ for RPHash-TWRP and $0.2007$ for Ward clustering. DStMM, with ARI $0.2989$, also compares favourably with the configurations reported by \citet{guyeux2019introducing}, including Isolation-Forest--preprocessed GMM ($0.2666$), Birch ($0.2820$), $K$-means ($0.2904$), and $K$-means++ ($0.2908$), as well as Extended-Isolation-Forest--preprocessed Birch ($0.2926$) and Ward clustering ($0.2972$). Beyond clustering accuracy, the implementation-based BIC criterion selects a DStMM with a non-trivial second mixture layer, $K_2=3$, rather than a candidate with a single second-layer component. Within the candidate grid considered here, the fitted model therefore retains additional heterogeneity through its deeper latent representation.

Therefore, the gas-sensor application reinforces the empirical pattern observed in the handwritten-digit example. For these high-dimensional data, which exhibit pronounced skewness and heavy tails, DStMM improves on both DGMM and RDMM while retaining a direct probabilistic model for the original sensor measurements. Within the implemented selection scheme, the retained architecture indicates that the fitted flexibility is expressed not only through the skew-$t$ component shape but also through a non-trivial second mixture layer.

\section{Discussion}
\label{sec:discussion}

The central statistical feature of the DStMM is the common positive scale propagated along a complete latent pathway. Conditional on that scale, every transition remains Gaussian; after the latent states are collapsed, the same variable simultaneously controls variance inflation and directional displacement in an exact GHST pathway distribution. Proposition~\ref{prop:path-collapse} and Corollary~\ref{cor:nested-models} therefore place the model between three familiar regimes: Gaussian deep mixtures, symmetric heavy-tailed deep mixtures, and single-layer skew-$t$ factor-analytic mixtures. The exact pathway likelihood is important because posterior path allocation and likelihood-based architecture comparison do not depend on a Monte Carlo approximation, even though stochastic simulation is used to update the local latent hierarchy. Section~\ref{sec:identifiability} separates pathway-level and local non-uniqueness. At the observation level, uniqueness of the positive-weight pathway mixture up to label permutation is treated as an explicit separation assumption. At the local level, the loading coordinates are not unique, and the reported implementation-based BIC uses the rotational correction in \eqref{eq:loading-ident-new} without an additional intermediate-coordinate gauge subtraction. The baseline active set $\mathcal A=\{1\}$ fixes the allocation of directional asymmetry to the observation transition in the reported DStMM fits; richer active sets can introduce further allocation ambiguities. The numerical results indicate when the additional skewness mechanism is most useful. Under symmetry, DStMM behaves similarly to RDMM, whereas stronger directional asymmetry produces progressively larger improvements in complete-pathway recovery, particularly when the common scale is more variable. This interaction is consistent with the model construction: the shared $W$ couples tail inflation and directional displacement, so the skewness term matters most when scale mixing is itself substantial. In both real-data examples the proposed model improves on the corresponding Gaussian and symmetric heavy-tailed deep alternatives. For the gas-sensor data, the implementation-based BIC criterion selects a DStMM with $K_2=3$, so the retained model includes a non-trivial second mixture layer within the candidate grid. Its ARI lies at the upper end of the selected published reference results in Table~\ref{tab:gas-comparison}(b), but those comparisons are contextual rather than controlled because preprocessing, dimensionality reduction, and optimisation differ across studies.

Several limitations remain. The complete-path enumeration grows as $|\mathcal S|=\prod_lK_l$, so computational cost can increase rapidly with depth even when local parameters are shared. The observed-data likelihood is multimodal, and with $M=1$ the stochastic updates are also subject to Monte Carlo variability; multiple starts and likelihood-based fit selection are therefore important in demanding applications, while a larger $M$ can be used when smoother Monte Carlo updates are required. The BIC accounting reported here follows the implemented diagonal-specific-covariance and rotational-degree convention; alternative covariance structures or different identification conventions would require a correspondingly revised parameter count. Architecture selection currently relies on a finite candidate grid for depth, mixture orders, latent dimensions, degrees-of-freedom pooling, and the skewness active set; joint or adaptive selection of these quantities remains an open problem. These limitations suggest several extensions: structured or partially pooled skewness across components and layers, scalable path-search or variational approximations for larger architectures, formal selection of the active set, and alternative asymmetric heavy-tailed kernels. The exact pathway representation provides a useful starting point for these developments because it separates the distributional theory at the observation level from the computational approximation used for the local latent states.

\appendix
\section*{Appendix}
\addcontentsline{toc}{section}{Appendix}

\section{Simulation diagnostics}
\label{app:simulation-diagnostics}

\setcounter{figure}{0}
\renewcommand{\thefigure}{\thesection.\arabic{figure}}
\setcounter{table}{0}
\renewcommand{\thetable}{\thesection.\arabic{table}}

\subsection{Experiment 2 parameter recovery}

\begin{figure}[!ht]
\centering
\begin{subfigure}[t]{0.65\linewidth}
    \centering
    \includegraphics[width=\linewidth]{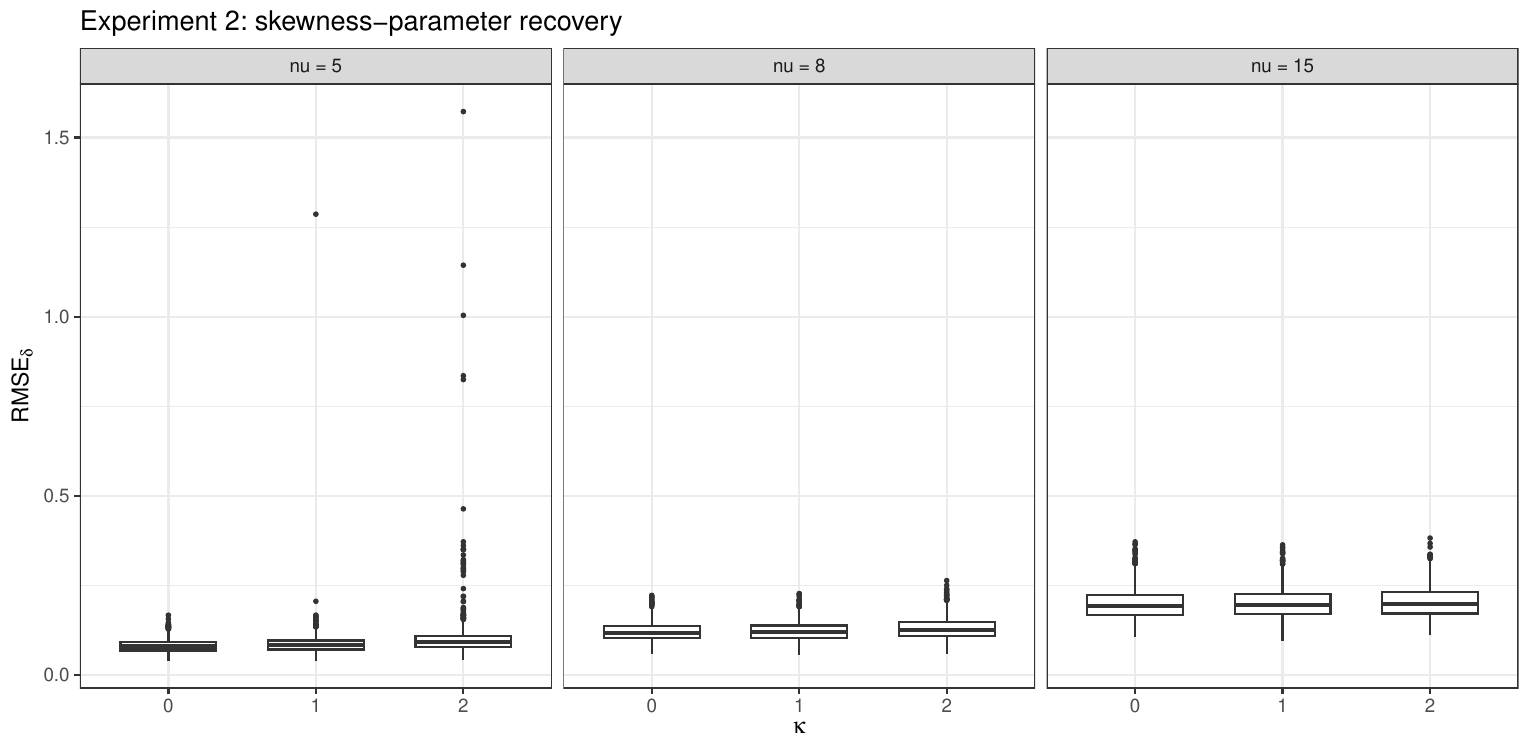}
    \caption{Skewness-vector RMSE.}
    \label{fig:exp2-delta-rmse}
\end{subfigure}

\medskip
\begin{subfigure}[t]{0.65\linewidth}
    \centering
    \includegraphics[width=\linewidth]{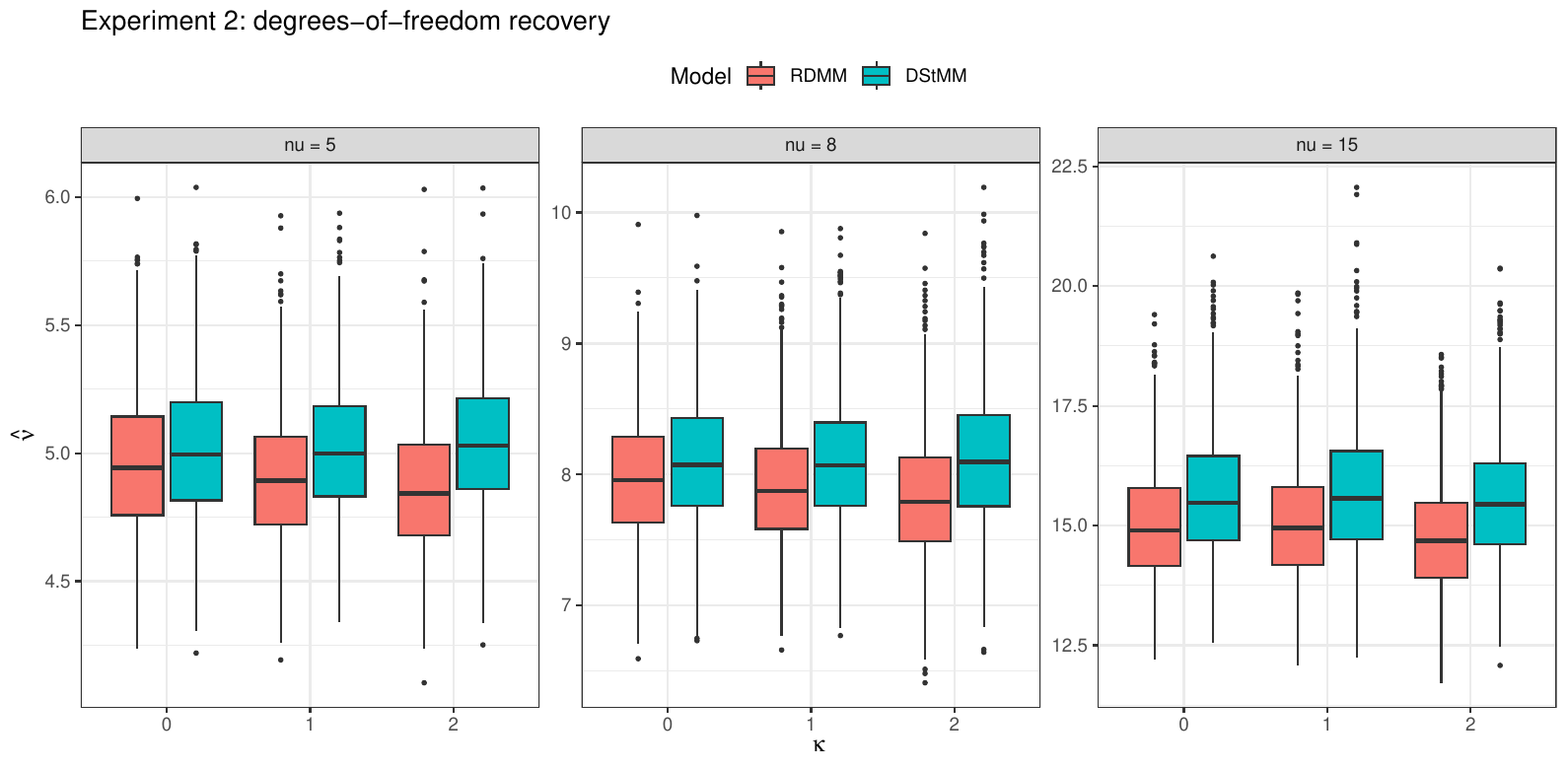}
    \caption{Recovery of $\nu$.}
    \label{fig:exp2-nu-recovery}
\end{subfigure}
\caption{Experiment 2 parameter-recovery diagnostics across the $(\nu,\kappa)$ factorial design.}
\label{fig:exp2-recovery-composite}
\end{figure}

The distributions in Figure~\ref{fig:exp2-recovery-composite} complement Table~\ref{tab:exp2-details}(b). As $\nu$ increases, the latent scale $W$ becomes more concentrated, and the recovery plots reflect the corresponding reduction in scale-mixing variation. Across these lighter-tailed settings, clustering performance remains strong while the parameter-recovery patterns continue to track the model's mean--variance structure.

\subsection{First-layer robustness}

\begin{figure}[tbp]
\centering
\begin{subfigure}[t]{0.76\linewidth}
    \centering
    \includegraphics[width=\linewidth]{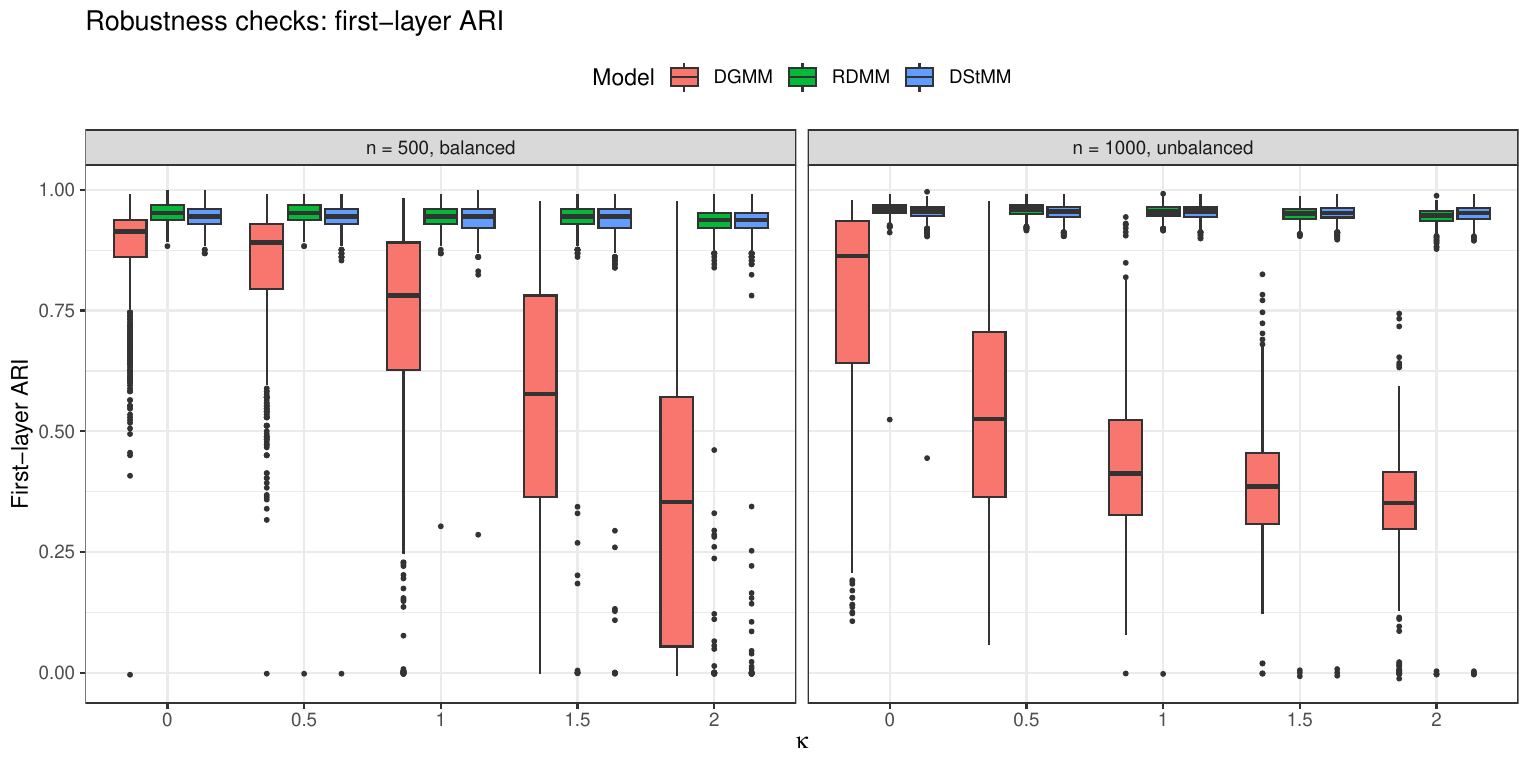}
    \caption{First-layer ARI.}
    \label{fig:rob-l1-ari}
\end{subfigure}

\medskip
\begin{subfigure}[t]{0.76\linewidth}
    \centering
    \includegraphics[width=\linewidth]{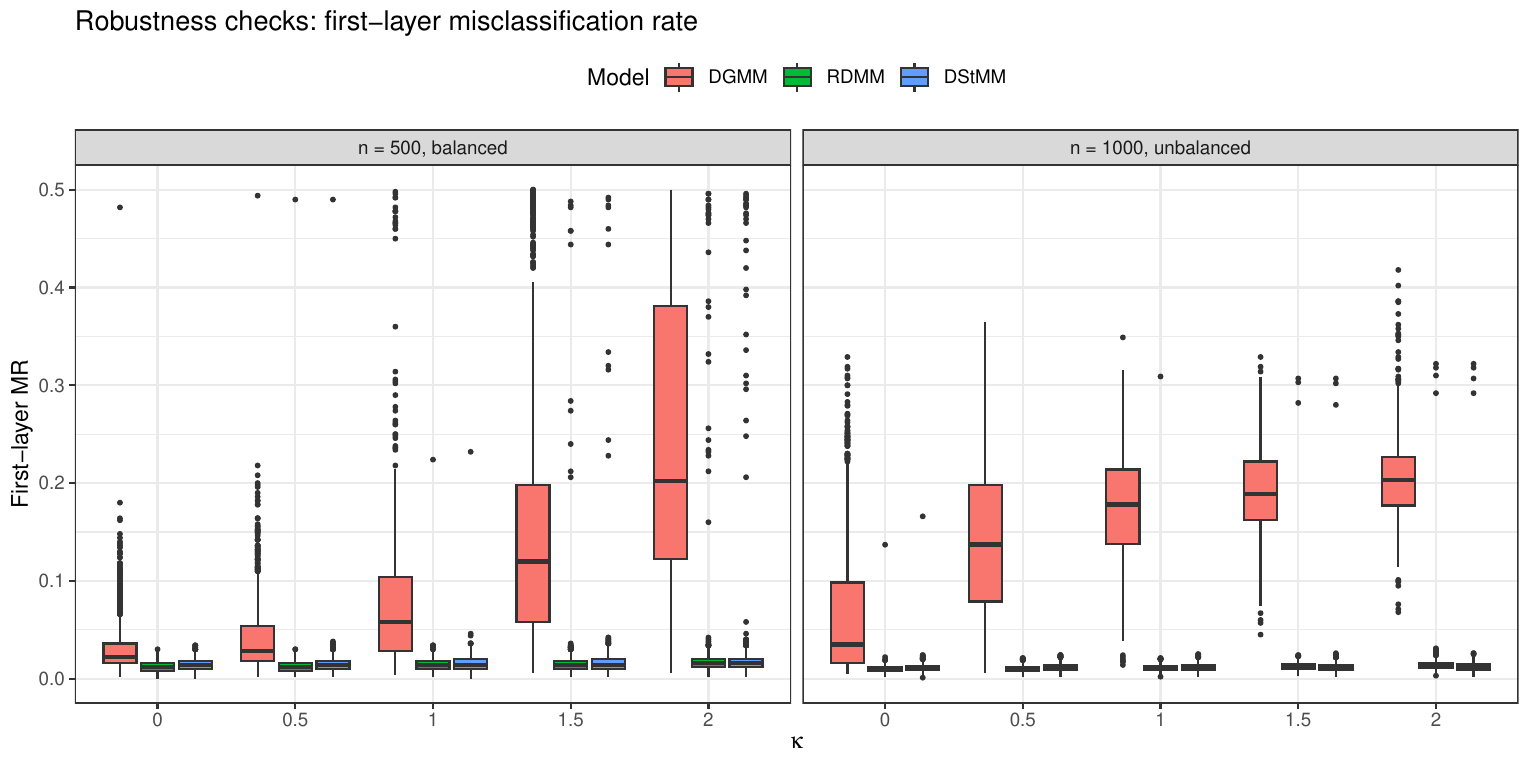}
    \caption{First-layer MR.}
    \label{fig:rob-l1-mr}
\end{subfigure}
\caption{Experiment 1 robustness checks at the first layer for the smaller-sample ($n=500$) and unbalanced-mixture designs.}
\label{fig:exp1-robustness-first-layer}
\end{figure}

Figure~\ref{fig:exp1-robustness-first-layer} complements Table~\ref{tab:exp1-robustness}. Across both robustness designs, RDMM and DStMM retain similar first-layer performance, reinforcing the conclusion that the main benefit of modelling skewness appears in complete-pathway discrimination rather than in the coarse first-layer partition.

\subsection{Computational profile and fitting stability}
\label{app:computational-diagnostics}

We summarise the computational profile and fitting stability using the same simulation runs. The timing summaries are based on successful fits.

\begin{figure}[tbp]
\centering
\includegraphics[width=0.65\linewidth]{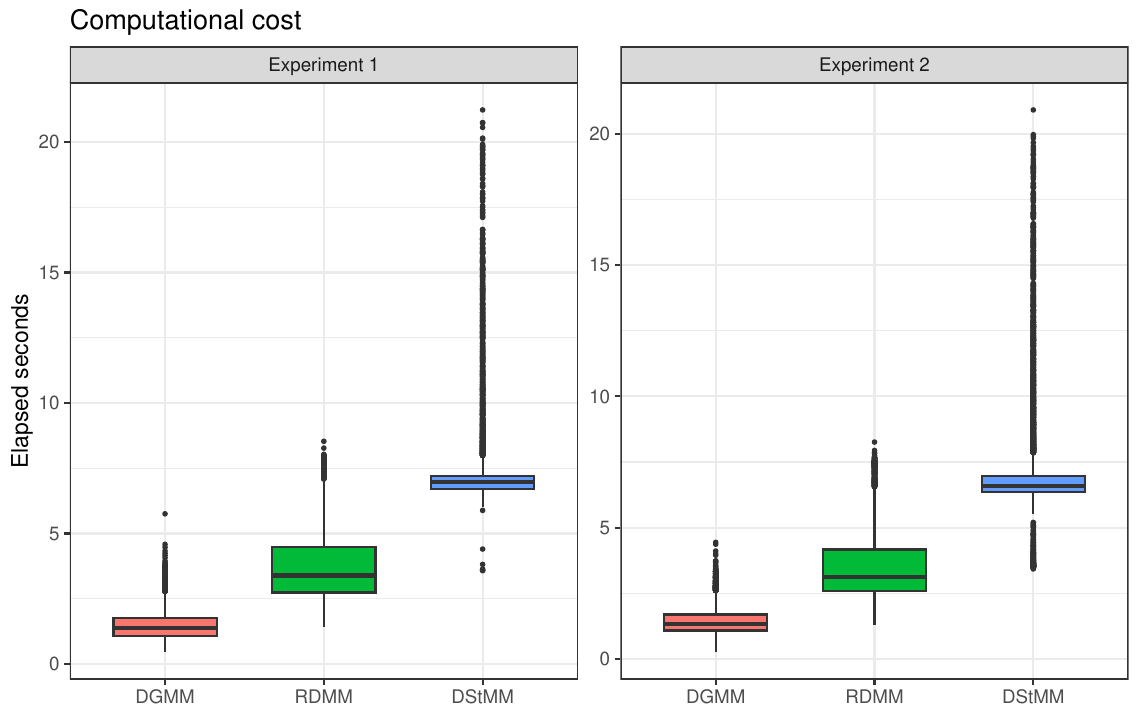}
\caption{Elapsed time per successful fit for DGMM, RDMM, and DStMM across the simulation experiments.}
\label{fig:runtime-comparison}
\end{figure}

Figure~\ref{fig:runtime-comparison} summarises the elapsed time per successful fit. In Experiment 1, the median times are 1.38 seconds for DGMM, 3.38 seconds for RDMM, and 6.97 seconds for DStMM; in Experiment 2, the corresponding medians are 1.33, 3.11, and 6.58 seconds. The DStMM timing profile reflects its Monte Carlo E-step together with the additional skewness updates, while remaining within single-digit seconds in these simulation settings.

Fitting stability is uniformly high. All 5000 Experiment 1 main fits succeed for each model. Across the 10,000 Experiment 1 robustness fits, DGMM, RDMM and DStMM have 9997, 10000 and 9999 successful fits, respectively. In Experiment 2, RDMM and DStMM both achieve 9000 successful fits out of 9000, while DGMM achieves 8988. These diagnostics show that DStMM retains excellent numerical stability across the simulation settings.

\FloatBarrier
\clearpage
\bibliographystyle{apalike}
\bibliography{refs}

\end{document}